\documentclass[12pt,english]{revtex4-2}
\usepackage{lmodern}

\usepackage[T1]{fontenc}
\usepackage[latin9]{inputenc}
\usepackage{babel}
\usepackage{dsfont}
\usepackage{amsmath}
\usepackage{amsthm}
\usepackage{amssymb}
\usepackage{graphicx}
\usepackage{geometry}
\usepackage{wasysym}
\usepackage[pdfusetitle,
 bookmarks=true,bookmarksnumbered=true,bookmarksopen=false,
 breaklinks=false,pdfborder={0 0 0},pdfborderstyle={},backref=false,colorlinks=false]
 {hyperref}

\makeatletter

\providecommand{\tabularnewline}{\\}

\theoremstyle{plain}
\newtheorem{prop}{\protect\propositionname}
\theoremstyle{plain}
\newtheorem{thm}{\protect\theoremname}
\theoremstyle{plain}
\newtheorem{cor}{\protect\corollaryname}
\theoremstyle{plain}
\newtheorem{lem}{\protect\lemmaname}

\makeatother

\providecommand{\corollaryname}{Corollary}
\providecommand{\lemmaname}{Lemma}
\providecommand{\propositionname}{Proposition}
\providecommand{\theoremname}{Theorem}

\begin{document}
\title{Series solutions to the TOV equations}
\author{Paulo Luz}
\email{paulo.luz@tecnico.ulisboa.pt}

\affiliation{Centro de Astrof\'isica e Gravita\c{c}\~ao - CENTRA, Departamento
de F\'isica, Instituto Superior T\'ecnico - IST, Universidade de
Lisboa - UL, Av. Rovisco Pais 1, 1049-001 Lisboa, Portugal}
\affiliation{Escola Superior N\'autica Infante D. Henrique, Pa\c{c}o de Arcos,
Portugal.}
\begin{abstract}
We present general series solutions to the Tolman--Oppenheimer--Volkoff
equations for compact stellar objects. We develop an algorithm to
compute the coefficients of the power series in terms of the equation
of state and its derivatives with respect to the thermodynamic variables.
Using these results, we establish general properties of analytic solutions
and their relation to the regularity of the equation of state. Applying
the theory of Pad\'e approximants, we derive series representations
for meromorphic functions whose domains of convergence may include
isolated poles. These analytic solutions are then used to obtain closed-form
expressions to approximate the radius and mass of stellar objects.
We apply the formalism to specific models, namely fluids with affine
equations of state and polytropic fluids, and compare the results
with those obtained from numerical integration. Lastly, we extend
the formalism to piecewise equations of state, deriving series solutions
that can be matched across transition hypersurfaces.
\end{abstract}
\maketitle

\section{Introduction}

Compact stellar objects are among the richest environments for probing
fundamental physics in regimes that are far beyond those presently
accessible to man-made laboratories.  Active and upcoming electromagnetic
and gravitational-wave detectors are expected to provide accurate
and statistically significant data on the macroscopic properties and
the evolution of massive compact objects. Relating the macroscopic
features of stellar objects to fundamental properties of the matter
fields that compose them will enable tighter constrains on theoretical
models and deepen our understanding of matter under extreme conditions~\citep{Huxford_et_al_2024,Lattimer_Prakash_2001,Webb_Barret_2007,Ozel_Psaltis_2009,Lindblom_2010,Oertel_Review_2017,Baym_et_al_report,Kojo_Baym_Hatsuda_2022}.

To a first approximation, compact stellar objects can be characterized
by static, spherically symmetric solutions of the theory of general
relativity. Under these symmetries and for a perfect fluid source,
the Einstein field equations can be cast in a simplified form: the
Tolman--Oppenheimer--Volkoff (TOV) equations. In the presence of
matter, the system is closed by specifying the fluid source, typically
in the form of an equation of state (EoS). Depending on the properties
of the matter fields, different solutions exist with markedly distinct
behavior.

Deriving an EoS from nuclear physics or quantum chromodynamics that
accurately describes matter under the conditions present in compact
objects interiors remains an extremely challenging task. Although
significant progress has been made in recent decades, a key difficulty
lies in the poorly constrained properties of matter at extreme densities,
particularly those expected to be found in neutron stars cores. This
uncertainty has motivated the development of numerous models based
on different assumptions about the state of matter~\citep{Akmal_Pandharipande_Ravenhall_1998,Morales_Pandharipande_Ravenhall_2002,Mukherjee_2009,Kruger_et_al_2013,Otto_Oertel_Schaefer_2020,Raduta_Oertel_Sedrakian_2020}.
Alternatively, extensive research has gone into constructing phenomenological
models that fit the EoS to available observational data (see, e.g.,
Refs.~\citep{Read_et_al_2009,Boyle_2020,Suleiman_et_al_2022} and
references therein). Nonetheless, to constrain the EoS, the theoretical
predictions of the TOV equations must be compared against experimental
data. In that regard, obtaining reliable solutions to the equations
is essential. 

Given the complexity of the TOV equations, very few exact solutions
suitable for stellar objects are known (see, e.g., Refs.~\citep{Delgaty_Lake_1998,MacCallum_Book_2009}).
For realistic EoS, numerical methods are often required to find approximate
solutions. However, this makes it difficult to extract general relationships
between stellar properties and the parameters of the EoS. In contrast,
analytic solutions can provide valuable qualitative insight into these
functional dependencies, which is crucial for connecting theoretical
predictions with observations. Moreover, the TOV equations form a
stiff system of differential equations with a singular point. Depending
on the EoS, their numerical integration is nontrivial, often requiring
working with very high precision or considering reformulations of
the original equations to control numerical errors~\citep{Lindblom_1998}.

It is then not surprising that significant attention has been directed
to find exact solutions of the TOV equations. In Refs.~\citep{Rahman_Visser_2002,Lake_2003,Martin_Visser_2004,Boonserm_Visser_Weinfurtner_2005,Boonserm_Visser_Weinfurtner_2007},
generating theorems were proposed to find solutions by deforming a
starting seed exact solution or by imposing a particular form for
one of the metric coefficients. In Ref.~\citep{Carloni_Vernieri_2018},
the TOV were cast in an explicit covariant form, allowing for the
derivation of new exact solutions. In Ref.~\citep{Harko_Mak_2016},
power series solutions to the TOV equations were developed for linear
and polytropic EoS. In Ref.~\citep{Saad_et_al_2017}, an analytic
solution was derived for polytropic fluids, where Pad\'e approximation
theory was used to extended the radius of convergence.

In this article, we derive series solutions to the TOV equations in
two general cases: when the energy density is a known analytic function,
and when a sufficiently regular barotropic equation of state is specified.
The article is organized as follows. In Section~\ref{Section:Reformulation_TOV},
we derive general series solutions to the TOV equations for a given
energy density profile. In Section~\ref{Section:General_Solution_EOS},
we consider the TOV equations coupled to a barotropic equation of
state, develop an algorithm to compute the power series coefficients
of the solution and establish general properties. The analytic results
are applied to particular types of fluids and compared with the numerical
approximations, in Section~\ref{Section:Applications_and_examples}.
In Section~\ref{Section:Piecewise}, we extend the formalism to piecewise
equations of state. In Section~\ref{Section:Conclusions}, we conclude.
In Appendix~\ref{Section:Appendix_Proof-of-Theorem}, we detail the
proof of the parity of the power series expansions of the energy density
and pressure about the center the star, for sufficiently regular equations
of state.

The mathematical formulas in the article are written in the geometrized
unit system, where $8\pi G=c=1$, and assuming the metric signature
$\left(-+++\right)$.

\section{Reformulation of the TOV equations\label{Section:Reformulation_TOV}}

\subsection{Equivalent differential system}

The TOV equations follow from the Einstein field equations considering
a static, spherically symmetric spacetime permeated by a perfect fluid
with energy density $\mu$ and pressure $p$. The original version
of the equations form a system of first order, non-linear differential
equations with a singular point. In Schwarzschild coordinates: $\left(t,r,\theta,\varphi\right)$,
where $r$ represents the circumferential radius, the TOV equations
read
\begin{align}
\frac{dp}{dr} & =-\frac{r^{2}\left(\mu+p\right)}{2\left(r-2M\right)}\left(p+\frac{2M}{r^{3}}-\Lambda\right)\,,\label{Reformulation_eq:Original_TOV_p}\\
\frac{dM}{dr} & =\frac{1}{2}\left(\mu+\Lambda\right)r^{2}\,,\label{Reformulation_eq:Original_TOV_M}
\end{align}
where $\Lambda$ represents the cosmological constant and
\begin{equation}
M\left(r\right)=\frac{1}{2}\int^{r}_{0}\left(\mu+\Lambda\right)x^{2}dx\,,\label{Reformulation_eq:Mass_function_definition}
\end{equation}
is dubbed the mass function. From hereon, we will omit the mass function
dependency on the radial coordinate. The system must be completed
by providing an equation that relates the energy density with the
pressure: an equation of state for the fluid, or an expression that
specifies one of the matter variables as a function of the spacetime
coordinates.

Solutions of TOV equations aim to model the interior of compact stellar
objects. As such, the contributions from the cosmological constant,
$\Lambda$, are often disregarded. We will consider this simplification
and throughout the article set $\Lambda=0$. Nonetheless, the inclusion
of the cosmological constant is straightforward and does not affect
the overall discussion below. Moreover, we will be interested in regular
solutions of the TOV equations, therefore, in what follows, we will
consider that $r>2M$ and call such solutions \emph{stars}.

Deriving exact solutions to the TOV equations for an equation of state
is a challenging task. However, various alternative strategies can
be employed to obtain solutions. One simplifying approach is to prescribe
the energy density directly as a function of spacetime. Following
this approach, various closed-form solutions have been found. To the
best of our knowledge, all known exact solutions share the property
that the energy density, $\mu\left(r\right)$, is real analytic at
the center of the star, at $r=0$. In this section, we then focus
on deriving general series solutions to the TOV equations given a
real analytic energy density. In the next section, we build on these
results to formulate an algorithm for computing power series solutions
to the TOV equations coupled with an equation of state.

We start by reformulating the system of equations to an equivalent
system where it is simpler to find the general power series solution
of the TOV equations for a know energy density function. Let the potentials
$\left(\phi,\mathcal{A},\mathcal{E}\right)$, such that
\begin{align}
\phi & =\frac{2}{r}\sqrt{1-\frac{2M}{r}}\,,\label{Reformulation_eq:phi_M_relation}\\
\mathcal{A}\phi & =p+\frac{2M}{r^{3}}\,,\label{Reformulation_eq:A_phi_p_relation}\\
\mathcal{E} & =\frac{1}{3}\mu-\frac{2M}{r^{3}}\,.\label{Reformulation_eq:E_mu_relation}
\end{align}
These potentials follow from the 1+1+2 semi-tetrad decomposition of
the spacetime manifold~\citep{Clarkson_Barrett_2003,Betschart_Clarkson_2004,Clarkson_2007},
having, thus, intrinsic physical meaning. In the considered setup,
$\phi$ represents the spatial expansion coefficient of the normalized
radial gradient, $\mathcal{A}$ represents the radial component of
the 4-acceleration of the elements of volume of the fluid, and $\mathcal{E}$
is the fully radial component of the electric part of the Weyl tensor,
characterizing tidal forces. We remark that in general $\phi$ diverges
at $r=0$, but $r\phi$ is regular. Moreover, regular solutions for
$p$ exist only if $\mathcal{A}$ and $\mathcal{E}$ vanish at the
center of the star.

Given a line element of the form
\begin{equation}
ds^{2}=-g_{tt}dt^{2}+g_{rr}dr^{2}+r^{2}d\Omega^{2}\,,\label{Reformulation_eq:General_line_element}
\end{equation}
where $d\Omega^{2}$ represents the line element of the unit 2-sphere,
and $g_{tt}$ and $g_{rr}$ are functions of the radial coordinate
only, we have the following relations
\begin{equation}
\begin{aligned}\phi & =\frac{2}{r\sqrt{g_{rr}}}\,,\\
\mathcal{A} & =\frac{1}{2g_{tt}\sqrt{g_{rr}}}\frac{dg_{tt}}{dr}\,.
\end{aligned}
\label{Reformulation_eq:phi_A_metric_relations}
\end{equation}

Using the TOV equations~(\ref{Reformulation_eq:Original_TOV_p})
and (\ref{Reformulation_eq:Original_TOV_M}), or starting directly
from the 1+1+2 structure equations, we find
\begin{equation}
\frac{d\mathcal{A}}{dr}=\frac{3\mathcal{E}}{r\phi}+\frac{1}{r}\mathcal{A}-\frac{2}{r\phi}\mathcal{A}^{2}\,.\label{Reformulation_eq:Riccati_A}
\end{equation}
Equation~(\ref{Reformulation_eq:Riccati_A}) is a Riccati equation
with a regular singular point at $r=0$. Applying a change of variables,
this equation can be written as a second-order linear ordinary differential
equation. Let $u$ such that
\begin{equation}
\frac{2\mathcal{A}}{r\phi}=\frac{1}{u}\frac{du}{dr}\,.\label{Reformulation_eq:A_u_variables_relation}
\end{equation}
From Eq.~(\ref{Reformulation_eq:phi_A_metric_relations}), the variable
$u\propto\sqrt{g_{tt}}$. Using Eq.~(\ref{Reformulation_eq:A_u_variables_relation})
in (\ref{Reformulation_eq:Riccati_A}), we find that the new variable
verifies the differential equation
\begin{equation}
\frac{d^{2}u}{dr^{2}}=\frac{6\mathcal{E}}{r^{2}\phi^{2}}u+\left(\frac{1}{r}+\frac{4\mu+6\mathcal{E}}{3r\phi^{2}}\right)\frac{du}{dr}\,.\label{Reformulation_eq:u_2nd_ODE}
\end{equation}

To find general power series solutions, we will apply an algorithm
described in Refs.~\citep{Coddington_Levinson_Book_1955,Coddington_Carlson_Book_1997},
also considered in Refs.~\citep{Luz_Carloni_2024a,Luz_Carloni_2024b,Luz_Carloni_2024c,Luz_Carloni_2026}
to find power series solutions for perturbations of stars. For this
purpose, it is useful to cast the second order differential equation~(\ref{Reformulation_eq:u_2nd_ODE})
in the form of a system of first-order differential equations. Introducing
$s=du/dr$, in matrix form we have
\begin{equation}
\frac{d\mathds{W}}{dr}=\left(\frac{1}{r}\left[\begin{array}{cc}
1 & 0\\
0 & 0
\end{array}\right]+\Theta\right)\mathds{W}\,,\label{Reformulation_eq:System_A_equation}
\end{equation}
where $\mathds{W}=\left[\begin{array}{cc}
s & u\end{array}\right]^{T}$ and
\begin{equation}
\Theta=\left[\begin{array}{cc}
{\displaystyle \frac{4\mu+6\mathcal{E}}{3r\phi^{2}}} & {\displaystyle \frac{6\mathcal{E}}{r^{2}\phi^{2}}}\\
1 & 0
\end{array}\right]\,.\label{Reformulation_eq:Theta_matrix}
\end{equation}
For $\mu$ analytic at $r=0$, the matrix $\Theta$ is analytic at
$r=0$.

Before proceeding, we note that a system for $p$, formally analogous
to (\ref{Reformulation_eq:System_A_equation}), can be obtained straightforwardly
by applying the same procedure directly to Eq.~(\ref{Reformulation_eq:Original_TOV_p}).
Since the pressure is an important quantity in determining the properties
of compact stellar objects, it is natural to consider a differential
system that explicitly relates $p$ to the energy density $\mu$.
However, as will be shown in the following section, to derive series
solutions to the TOV equations coupled with an equation of state,
it is advantageous to use the variables $\mu$, $p$ and $\mathcal{A}$.
Consequently, the specific choice of variable, between $p$ or $\mathcal{A}$,
in the differential system is not essential. Moreover, solutions of
Eq.~(\ref{Reformulation_eq:System_A_equation}) immediately yield
the $g_{tt}$ component of the metric tensor, up to a multiplicative
constant.

\subsection{Power series solutions given an analytic energy density}

The first-order linear differential system~(\ref{Reformulation_eq:System_A_equation})
contains a singular point at $r=0$. To derive the general power series
solutions around that point, we will employ the Coddington-Levinson
procedure. This approach allow us to find the power series solutions
around the singular point, which may or may not be regular at the
singular point. Nevertheless, we will show that in the considered
setup all analytic solutions are regular at $r=0$.

The singular part of Eq.~(\ref{Reformulation_eq:System_A_equation})
is characterized by a constant diagonal matrix with constant eigenvalues:
0 and 1. This property allow us to derive the general formal solution.
Using the Coddington-Levinson algorithm, the analytic solutions of
the TOV equation are given by
\begin{align}
\left[\begin{array}{r}
s\\
u
\end{array}\right] & =\left[\begin{array}{lr}
r & 0\\
0 & 1
\end{array}\right]\mathds{P}\left(r\right)\left[\begin{aligned}c_{1}\\
c_{2}
\end{aligned}
\right]\,,\label{Reformulation_eq:Solution_s_u}\\
\mathcal{A} & =\frac{r\phi s}{2u}\,,\label{Reformulation_eq:Solution_s_u_A}\\
p & =\mathcal{E}+\mathcal{A}\phi-\frac{1}{3}\mu\,,\label{Reformulation_eq:Solution_s_u_p}
\end{align}
where $c_{1}$ and $c_{2}$ are integration constants, and the matrix
$\mathds{P}$ has power series about $r=0$
\begin{equation}
\mathds{P}\left(r\right)=\sum^{+\infty}_{m=0}\mathds{P}_{m}\,r^{m}\,,\label{Reformulation_eq:Power_series_P_matrix}
\end{equation}
with matrix coefficients $\mathds{P}_{m}$ given by the recurrence
relation
\begin{equation}
\begin{aligned}\mathds{P}_{0} & =\mathds{I}_{2}\,,\\
\mathds{P}_{m} & =\frac{1}{m}\sum^{m-1}_{k=0}\widetilde{\Theta}_{m-1-k}\,\mathds{P}_{k}\,,\text{for }\ensuremath{m\geq}1\,,
\end{aligned}
\label{Reformulation_eq:Recurrence_relation_P}
\end{equation}
where $\mathds{I}_{2}$ is the identity matrix of size 2 and
\begin{equation}
\widetilde{\Theta}=\left[\begin{array}{cc}
{\displaystyle \frac{r\left(\mu-\frac{2M}{r^{3}}\right)}{2\left(1-\frac{2M}{r}\right)}} & \,{\displaystyle \frac{\mu-\frac{6M}{r^{3}}}{2r\left(1-\frac{2M}{r}\right)}}\\
r & 0
\end{array}\right]\,.\label{Reformulation_eq:Theta_tilde}
\end{equation}
In Eq.~(\ref{Reformulation_eq:Recurrence_relation_P}), the terms
$\widetilde{\Theta}_{m}$ represent the matrix coefficients of the
power series expansion of the matrix $\widetilde{\Theta}$ about $r=0$,
that is $\widetilde{\Theta}\left(r\right)=\sum^{+\infty}_{m=0}\widetilde{\Theta}_{m}\,r^{m}$.

Since $\mathds{P}_{0}=\mathds{I}_{2}$, the integration constants
$c_{1}$ and $c_{2}$ can be directly related to the boundary conditions
at the center of the star. The important physical quantity is the
ratio $c_{1}/c_{2}$ and not their separate values, other than being
zero. For instance, setting $c_{2}=1$, we find that $c_{1}$ is given
in terms of the central energy density, $\mu_{c}$, and the central
pressure, $p_{c}$, as
\begin{equation}
c_{1}=\frac{1}{6}\left(3p_{c}+\mu_{c}\right)\,.\label{Reformulation_eq:Integration_constant_c1}
\end{equation}

For a real analytic function $\mu$, the radius of convergence of
the power series representation of $\mathds{W}$ can be related to
the radius of convergence of the power series representation of $\mu$
and the fixed points of $2M$. We summarize the result as follows:
\begin{prop}
\label{Proposition:Radius_Analytic_solutions} Let the energy density
be a real analytic function at $r=0$, with radius of convergence
$R$. For real values of the central density and central pressure,
$\mu_{c}$ and $p_{c}$, the TOV equations admit real analytic solutions
at $r=0$ for compact stellar objects. Given the non-zero fixed points
of $2M\left(r\right)$, $\left\{ r_{i}\right\} $, with $i=1,...,n$,
the radius of convergence of the power series representation of $\mathds{W}$,
$\bar{R}$, is given by the smallest, positive, modulus of the fixed
points of $2M\left(r\right)$ or $R$, that is $\bar{R}=\min_{i=1,...,n}\left\{ \left|r_{i}\right|,R\right\} $.
\end{prop}
\begin{proof}
The formal solution~(\ref{Reformulation_eq:Solution_s_u}) follows
from the Coddington-Levinson algorithm~\citep{Coddington_Levinson_Book_1955}.
The proofs in Refs.~\citep{Coddington_Levinson_Book_1955,Coddington_Carlson_Book_1997}
show that the resulting power series solution has radius of convergence
equal to the radius of convergence of $\Theta$.

Let $\mu$ be real analytic at $r=0$, such that $M\sim\mathcal{O}\left(r^{3}\right)$.
In particular, this implies that $M\left(r\right)/r=0$ at $r=0$.
Then, from Eqs.~(\ref{Reformulation_eq:Mass_function_definition})--(\ref{Reformulation_eq:E_mu_relation}),
$\Theta$ is real analytic at $r=0$. The radius of convergence of
$\Theta$ is either given by the smallest modulus of the non-zero
fixed points of $2M\left(r\right)$ or the radius of convergence of
the power series expansion of $\mu$, depending on which one is smaller.
\end{proof}

Proposition~\ref{Proposition:Radius_Analytic_solutions} allows us
to determine the disk of convergence of the power series expansion
of $u$ about $r=0$. Naturally, the radius of the disk might be smaller
than the radius of the star. However, it is possible to find series
solutions for $\mathcal{A}$ and $p$ that may converge in regions
beyond that disk.

Given the power series solutions for $u$, we can use Eqs.~(\ref{Reformulation_eq:Solution_s_u_A})
and (\ref{Reformulation_eq:Solution_s_u_p}) to find series solutions
for the potentials $\mathcal{A}$ and $p$. The power series representation
of $u$ can be used to find, for instance, the Maclaurin series of
$\mathcal{A}$ and of $p$. However, we are not limited to use Maclaurin
series. Alternatively, we may consider using Pad\'e approximants
centered at $r=0$ to expand $\mathcal{A}$ and $p$ by rational functions.
Since rational functions are meromorphic, Pad\'e approximants may
have radius of convergence greater than the radius of convergence
of the Maclaurin series of $u$.

Moreover, using Eqs.~(\ref{Reformulation_eq:Solution_s_u_A}) and
(\ref{Reformulation_eq:Solution_s_u_p}), we may directly express
$\mathcal{A}$ and $p$ as the ratio of power series. Truncating the
series in the numerator and denominator to a given order does not,
in general, yield a Pad\'e approximant of corresponding order since
it may not verify the same conditions on the higher-order derivatives.
However, the result also does not necessarily have radius of convergence
equal to the Maclaurin series of $u$.

We will further this discussion in the following sections.

\subsection{Examples\label{subsec:Examples_given_mu}}

Specifying a real analytic energy density, we can find series solutions
for the pressure and for all quantities that characterize the spacetime's
geometry, using Eqs.~(\ref{Reformulation_eq:Solution_s_u})--(\ref{Reformulation_eq:Integration_constant_c1}).
To exemplify the application of the results in the previous subsection,
we will consider the particular cases of the exact Tolman IV solution~\citep{Tolman_1939},
the Heintzmann IIa solution~\citep{Heintzmann_1969}, and a numerical
solution of the TOV equations with a rational energy density.

Assuming a line element~(\ref{Reformulation_eq:General_line_element}),
the Tolman IV solution is characterized by the non-trivial metric
coefficients
\begin{equation}
\begin{aligned}g_{tt}= & B^{2}\left(1+\frac{r^{2}}{A^{2}}\right)\,,\\
g_{rr}= & \frac{1+2\frac{r^{2}}{A^{2}}}{\left(1+\frac{r^{2}}{A^{2}}\right)\left(1-\frac{r^{2}}{R^{2}}\right)}\,,
\end{aligned}
\label{Reformulation_eq:TolmanIV_metric}
\end{equation}
where $A$, $B$ and $R$ are assumed to be non-zero, real constants.
The energy density and the pressure are given by
\begin{equation}
\begin{aligned}\mu & =\frac{3A^{4}+A^{2}\left(7r^{2}+3R^{2}\right)+2r^{2}\left(3r^{2}+R^{2}\right)}{R^{2}\left(A^{2}+2r^{2}\right)^{2}}\,,\\
p & =\frac{R^{2}-A^{2}-3r^{2}}{R^{2}\left(A^{2}+2r^{2}\right)}\,.
\end{aligned}
\label{Reformulation_eq:TolmanIV_mu_p}
\end{equation}

The interior of a star can be modeled by the Tolman IV solution. Let
the radius of the star, $r_{b}$, be defined as the value of the circumferential
radius of the surface at which the pressure vanishes, that is 
\begin{equation}
p\left(r_{b}\right)=0\:.\label{Reformulation_eq:radius_star_definition}
\end{equation}
For the Tolman IV solution, we find
\begin{equation}
r_{b}=\sqrt{\frac{R^{2}-A^{2}}{3}}\,.
\end{equation}

From Eqs.~(\ref{Reformulation_eq:TolmanIV_metric}) and (\ref{Reformulation_eq:TolmanIV_mu_p}),
we see that the energy density and the pressure have singular points
at $r=\pm iA/\sqrt{2}$, and $\sqrt{g_{tt}}$ has singular points
at $r=\pm iA$. Therefore, depending on the values of the parameters
$R$ and $A$, the radius of convergence of the Maclaurin series of
$u$ and $p$, might be smaller than the radius of the star. However,
$p$ is a rational function, such that both the numerator and the
denominator are quadratic polynomials. Therefore, it is equal to its
Pad\'e approximant of order $\left[2,2\right]$ about $r=0$. That
is, using Eqs.~(\ref{Reformulation_eq:Solution_s_u})--(\ref{Reformulation_eq:Integration_constant_c1})
to find the series expansion to order 4, we can, in this case, find
the exact expression for the pressure.

The Tolman IV solution is simple enough that it is possible to find
exact expressions for all the thermodynamic variables, radius of the
star and the singular points of the various quantities. If the pressure
can be expressed as a rational function, we can find an exact expression
from Eqs.~(\ref{Reformulation_eq:Solution_s_u})--(\ref{Reformulation_eq:Integration_constant_c1}),
using Pad\'e approximant theory. In general this is not the case.

Consider the exact solution first derived by Heintzmann, such that
the nontrivial metric components in Schwarzschild coordinates are
given by
\begin{equation}
\begin{aligned}g_{tt} & =A^{2}\left(ar^{2}+1\right)^{3}\,,\\
g_{rr} & =\left(1-3ar^{2}\frac{C\left(4ar^{2}+1\right)^{-\frac{1}{2}}+1}{2\left(ar^{2}+1\right)}\right)^{-1}\,,
\end{aligned}
\label{Reformulation_eq:Heintzmann_metric}
\end{equation}
with parameters $A$, $a$ and $C$ . The energy density and the pressure
of the fluid source are given by
\begin{equation}
\begin{aligned}\mu & =\frac{3a\left[\left(4a^{2}r^{4}+13ar^{2}+3\right)\sqrt{4ar^{2}+1}+C\left(9ar^{2}+3\right)\right]}{2\left(ar^{2}+1\right)^{2}\left(4ar^{2}+1\right)^{\frac{3}{2}}}\,,\\
p & =-\frac{3a\left(7aCr^{2}+3\left(ar^{2}-1\right)\sqrt{4ar^{2}+1}+C\right)}{2\left(ar^{2}+1\right)^{2}\sqrt{4ar^{2}+1}}\,.
\end{aligned}
\label{Reformulation_eq:Heintzmann_energy_pressure}
\end{equation}

Depending on the values of $a$ and $C$, $\sqrt{g_{tt}}$ and $p$
might be singular at points in the complex plane with modulus smaller
than the radius of the star. The singular points of $\sqrt{g_{tt}}$
are solutions to $ar^{2}=-1$, and the singular points $p$ are solutions
to $ar^{2}=-1$ or to $4ar^{2}=-1$. The Heintzmann IIa solution provides
an example where using Pad\'e approximants allow us to locally approximate
the solution to points outside the disk of convergence of the Maclaurin
series, and outside the disk of convergence of the series ratio solutions
following from Eqs.~(\ref{Reformulation_eq:Solution_s_u})--(\ref{Reformulation_eq:Solution_s_u_p}).

Figure~\ref{Figure:Pressure_profile_HeintzmannIIa} shows the behavior
of the pressure profile, its Taylor polynomial at $r=0$ truncated
to a given order, and the truncated series ratio solution following
from Eqs.~(\ref{Reformulation_eq:Solution_s_u})--(\ref{Reformulation_eq:Integration_constant_c1}).
In that particular realization of the Heintzmann IIa solution, $\sqrt{g_{tt}}$
and $p$ have singular points with modulus smaller than the radius
of the star. Indeed, in that case, both the radius of convergence
of the Maclaurin series of the exact solution, and the radius of convergence
of the series ratio solutions following from Eqs.~(\ref{Reformulation_eq:Solution_s_u})--(\ref{Reformulation_eq:Solution_s_u_p}),
are smaller than the radius of the star. Therefore, neither can be
used to accurately describe the pressure for all points in the interior
of the star. However, the pressure has only isolated singular points
and can be approximated by a rational function. For instance, the
diagonal Pad\'e approximant of order $\left[8,8\right]$, for the
spacetime parameters considered for Figure~\ref{Figure:Pressure_profile_HeintzmannIIa},
differs pointwise from the exact values of the normalized pressure,
for any point within the interior of the star, no more than $3.6\times10^{-4}$.

\begin{figure}
\includegraphics[totalheight=0.16\paperheight]{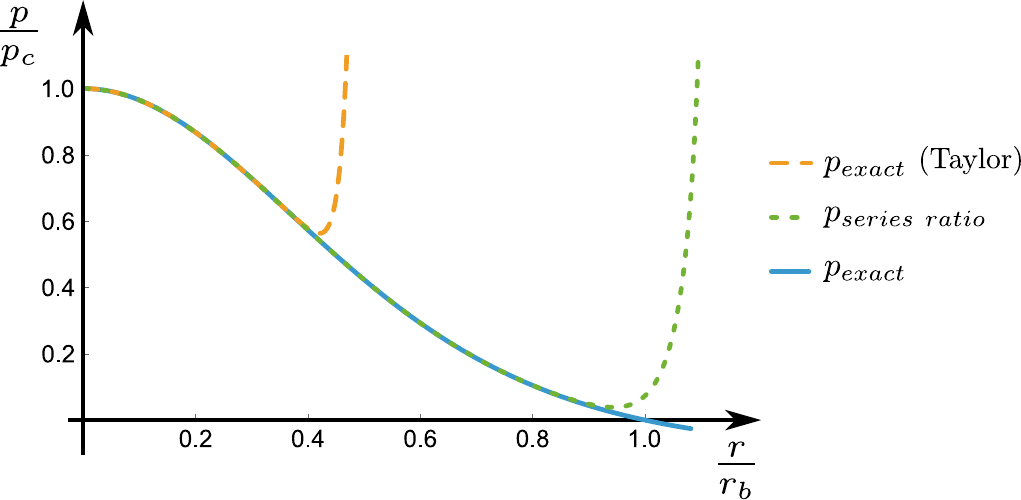}\caption{\label{Figure:Pressure_profile_HeintzmannIIa}Pressure profile of
the Heintzmann solution~(\ref{Reformulation_eq:Heintzmann_energy_pressure})
assuming the parameters $A=1$, $a=16$ and $C=-\frac{1}{3}$, its
Taylor polynomial about $r=0$ truncated to order 30, and the series
ratio solution from Eqs.~(\ref{Reformulation_eq:Solution_s_u})--(\ref{Reformulation_eq:Integration_constant_c1})
truncated to order 30.}

\end{figure}

For the Tolman IV and Heintzmann IIa spacetimes, the energy density
is given by an exact expression. However, Eqs.~(\ref{Reformulation_eq:Solution_s_u})--(\ref{Reformulation_eq:Integration_constant_c1})
can be used to approximate solutions to the TOV equations for an analytic
energy density, that cannot be expressed in exact form and can only
be approximated through numerical methods.

The TOV equations form a very stiff system of differential equations.
Moreover, Eq.~(\ref{Reformulation_eq:Original_TOV_p}) has a singular
point at $r=0$, such that numerical approximations might require
high working precision. Nonetheless, a fundamental limitation of numerical
solutions is that it might be very difficult to determine general
dependencies of the solutions on the parameters. In that regard, considering
series solutions, it is possible to establish general analytic expressions
to estimate the properties of the solutions, considering only a few
terms of the series.

To exemplify this idea, consider a solution of the TOV equations with
energy density
\begin{equation}
\mu=a+\frac{b}{1-cr^{2}}\,.\label{Reformulation_eq:numerical_example_mu_rational}
\end{equation}
Using Eqs.~(\ref{Reformulation_eq:Solution_s_u})--(\ref{Reformulation_eq:Integration_constant_c1}),
the Pad\'e approximant of order $\left[2,2\right]$ for the pressure
yields a rational function where both numerator and denominator are
polynomials of order 2. Computing its positive roots, we find the
analytic approximation for the radius of the star
\begin{equation}
r_{b}\approx\sqrt{\frac{60p_{c}}{5\left(a+b\right)^{2}+2\left(a+b\right)\left[10p_{c}+3bc\left(\frac{6}{a+b}-\frac{5}{a+b+p_{c}}-\frac{1}{a+b+3p_{c}}\right)\right]}}\,.\label{Reformulation_eq:numerical_example_rb_analytic}
\end{equation}

Table~\ref{Table:Numerical_example_comparison_radius_approximations}
presents a comparison between the circumferential radii obtained by
numerically integrating the TOV equations and those predicted by the
analytic approximation~(\ref{Reformulation_eq:numerical_example_rb_analytic}),
for stars with energy density function~(\ref{Reformulation_eq:numerical_example_mu_rational}).
In the example, for lower values of the compactness parameter, the
analytic approximation closely matches the exact value. For larger
compactness values, the analytic approximation significantly deviates
from the numeric result and a higher order analytic approximation
should be used. We note, however, that such stars become dynamically
unstable at higher values of the compactness parameter. For instance,
keeping $a=1$, $b=-\frac{1}{5}$ and $c=\frac{1}{4}$, and considering
$p_{c}\apprge0.32$, such that $M/r_{b}\apprge0.29$, the star is
unstable under linear, adiabatic radial perturbations.

\begin{table}
\begin{tabular}{|c|c|c|c|c|}
\hline 
 & {\small$\frac{M}{r_{b}}$} & {\small ~~$r_{b}\,\left(\text{Num.}\right)$~~} & {\small ~~$r_{b}\,\left(\text{An. approx.}\right)$~~} & ~~~~~~$\delta$~~~~~~\tabularnewline
\hline 
\hline 
{\small ~~$p_{c}=0.05$~~} & {\small ~~0.101~~} & {\small 0.883} & {\small 0.881} & {\small 0.3\%}\tabularnewline
\hline 
{\small$p_{c}=0.10$} & {\small 0.167} & {\small 1.159} & {\small 1.148} & {\small 0.9\%}\tabularnewline
\hline 
{\small$p_{c}=0.15$} & {\small 0.213} & {\small 1.333} & {\small 1.309} & {\small 1.8\%}\tabularnewline
\hline 
{\small$p_{c}=0.20$} & {\small 0.247} & {\small 1.461} & {\small 1.418} & {\small 3.0\%}\tabularnewline
\hline 
{\small$p_{c}=0.25$} & {\small 0.271} & {\small 1.566} & {\small 1.497} & {\small 4.5\%}\tabularnewline
\hline 
{\small$p_{c}=0.30$} & {\small 0.288} & {\small 1.668} & {\small 1.556} & {\small 6.7\%}\tabularnewline
\hline 
\end{tabular}

\caption{\label{Table:Numerical_example_comparison_radius_approximations}
Values of the circumferential radii of stars with energy density function~(\ref{Reformulation_eq:numerical_example_mu_rational})
for different values of the central pressure. In all cases $a=1$,
$b=-\frac{1}{5}$ and $c=\frac{1}{4}$. In the second column it is
indicated the value of the compactness parameter of the star; in the
third column, the value of $r_{b}$ found from the numerical integration;
in the fourth column it is presented the value of $r_{b}$ found from
the analytic approximation~(\ref{Reformulation_eq:numerical_example_rb_analytic});
in the last column is the percent error, $\delta$, for the analytic
approximation of $r_{b}$ relative to the numeric approximation.}
\end{table}

\section{Series Solutions to the TOV with an Equation of State\label{Section:General_Solution_EOS}}

The results in the previous section allow us to compute general series
solutions to the TOV equations, if the energy density is a known real
analytic function at $r=0$. However, physically, the energy density
is not a known function of the spacetime \emph{ab initio}. Solutions
for compact stellar objects are found by coupling the TOV equations
with a barotropic equation of state for the fluid source, relating
the pressure with the energy density. In this section, assuming a
sufficiently regular barotropic EoS, we derive series solutions for
the spacetime potentials.

Let the energy density and the pressure of the fluid source be related
by an EoS of the form
\begin{equation}
f\left(\mu,p\right)=0\,,\label{General_Power_Series_eq:EoS_general}
\end{equation}
where $f$ is an analytic function of the matter variables at the
center of the star, and $\partial_{\mu}f\neq0$, also at the center
of the star. Then, the spacetime is characterized by the TOV equations
coupled with the constraint equation~(\ref{General_Power_Series_eq:EoS_general}).

Using the chain rule, Eq.~(\ref{General_Power_Series_eq:EoS_general})
implies
\begin{equation}
f_{\mu}\partial_{r}\mu+f_{p}\partial_{r}p=0\,,\label{General_Power_Series_eq:Derivative_EoS}
\end{equation}
where the subscript in $f$ represents the partial derivative with
respect to the indicated thermodynamic variable. Using Eqs.~(\ref{Reformulation_eq:Original_TOV_p})--(\ref{Reformulation_eq:A_phi_p_relation})
we have
\begin{equation}
\frac{dp}{dr}=-\frac{2\left(\mu+p\right)\mathcal{A}}{r\phi}\,.\label{General_Power_Series_eq:1p1p2_TOV_pressure}
\end{equation}
Combining Eqs.~(\ref{General_Power_Series_eq:Derivative_EoS}) and
(\ref{General_Power_Series_eq:1p1p2_TOV_pressure}) yields the important
relation
\begin{equation}
\frac{d\mu}{dr}=\frac{2\left(\mu+p\right)f_{p}}{r\phi f_{\mu}}\mathcal{A}\,.\label{General_Power_Series_eq:EoS_derivative_mu}
\end{equation}

We remark that defining the square of the adiabatic speed of sound
as $c^{2}_{s}\equiv\left(dp/d\mu\right)_{S}$, where the derivative
is taken at constant entropy, we have $c^{2}_{s}=-f_{\mu}/f_{p}$.
This implies that the adiabatic index of a perfect fluid, $\gamma$,
is given by
\begin{equation}
\gamma=-\frac{\left(\mu+p\right)f_{\mu}}{pf_{p}}\,.\label{General_Power_Series_eq:Adiabatic_index}
\end{equation}
Therefore, Eq.~(\ref{General_Power_Series_eq:EoS_derivative_mu})
can be equivalently written in terms of the adiabatic speed of sound
or the adiabatic index.

Continuing, Eq.~(\ref{General_Power_Series_eq:EoS_derivative_mu})
relates the derivative of the energy density directly with $\mu$,
$p$, $\mathcal{A}$ and the derivatives of $f$. Following the results
in Appendix~\ref{Section:Appendix_Proof-of-Theorem}, if $\mu$,
$p$, $\mathcal{A}$ and $f$ are analytic at $r=0$, taking the derivative
of order $n$ of (\ref{General_Power_Series_eq:EoS_derivative_mu})
implies that the coefficient of the $\left(n+1\right)^{th}$-order
term of the power series of $\mu$ about $r=0$ depends, at most,
on the $n^{th}$-order derivatives of $\mu$, $p$ and $\mathcal{A}$
at the center. In turn, from Eq.~(\ref{Reformulation_eq:Solution_s_u})--(\ref{Reformulation_eq:Integration_constant_c1}),
for $n>0$, the coefficients of the $n^{th}$-order terms of the power
series of $\mathcal{A}$ and of $p$ depend at most on the derivative
$\mu^{\left(n-1\right)}\left(0\right)$. Then, provided the central
energy density, $\mu_{c}$, and the central pressure, $p_{c}$, we
can use Eqs.~(\ref{Reformulation_eq:Solution_s_u})--(\ref{Reformulation_eq:Integration_constant_c1})
together with Eq.~(\ref{General_Power_Series_eq:EoS_derivative_mu})
and its derivatives to compute the coefficients of the power series
of $\mu$, $p$ and $\mathcal{A}$ recursively.

Using this procedure, the coefficients of the power series of the
potentials can be computed symbolically to arbitrary order, for a
general barotropic EoS or for a specific family of EoS. Below, we
present the first terms of the power series of $\mu$, $p$ and $\mathcal{A}$
about $r=0$, up to order 4. Expectedly, in the general case, the
expressions for the higher order terms become rapidly quite large
due to the number of higher order derivatives of the function $f$.

\begin{equation}
\begin{aligned}\mu\left(r\right) & =\mu_{c}+{\displaystyle \frac{\left(\mu_{c}+p_{c}\right)\left(\mu_{c}+3p_{c}\right)f_{p}}{12f_{\mu}}}r^{2}\\
 & +\frac{\mu^{2}_{c}+3p^{2}_{c}+4\mu_{c}p_{c}}{48f_{\mu}}\left[\frac{\left(4\mu_{c}+9p_{c}\right)\left(f_{p}\right)^{2}}{15f_{\mu}}-\frac{\left(\mu_{c}+p_{c}\right)\left(\mu_{c}+3p_{c}\right)f_{pp}}{6}\right.\\
 & \left.+\frac{\left(\mu_{c}+p_{c}\right)\left(\mu_{c}+3p_{c}\right)f_{p}f_{\mu p}}{3f_{\mu}}-\frac{\left(\mu_{c}+p_{c}\right)\left(\mu_{c}+3p_{c}\right)\left(f_{p}\right)^{2}f_{\mu\mu}}{6\left(f_{\mu}\right)^{2}}-p_{c}f_{p}\right]r^{4}+\mathcal{O}\left(r^{6}\right)\,,\\
p\left(r\right) & =p_{c}{\displaystyle -\frac{\left(\mu_{c}+p_{c}\right)\left(\mu_{c}+3p_{c}\right)}{12}}r^{2}{\displaystyle -\frac{\left(\mu_{c}+p_{c}\right)\left(\mu_{c}+3p_{c}\right)\left[\left(4\mu_{c}+9p_{c}\right)f_{p}-15p_{c}f_{\mu}\right]}{720f_{\mu}}}r^{4}+\mathcal{O}\left(r^{6}\right)\,,\\
\mathcal{A}\left(r\right) & ={\displaystyle \frac{\mu_{c}+3p_{c}}{6}}r+{\displaystyle \frac{\left(\mu_{c}+3p_{c}\right)\left[3\left(\mu_{c}+p_{c}\right)f_{p}-5\left(\mu_{c}+3p_{c}\right)f_{\mu}\right]}{360f_{\mu}}}r^{3}+\mathcal{O}\left(r^{5}\right)\,,
\end{aligned}
\label{General_Power_Series_eq:General_series}
\end{equation}
where the derivatives of the function $f$ are evaluated at $\left(\mu_{c},p_{c}\right)$.

As previously discussed, regularity of solutions of the TOV equations
require that $\mathcal{A}$ vanishes at $r=0$. In turn, this implies
that $\partial_{r}p=0$ at the center of the star. Moreover, if $f_{\mu}\neq0$
at the center of the star, it follows from Eq.~(\ref{General_Power_Series_eq:EoS_derivative_mu})
that $\partial_{r}\mu=0$. Indeed, the expressions in Eq.~(\ref{General_Power_Series_eq:General_series})
suggest that the energy density and the pressure are even functions
of the radial coordinate, whereas $\mathcal{A}$ is an odd function.
The following result establishes that this is a consequence of the
regularity of the EoS at the center of the star.
\begin{thm}
\label{Theorem:Even_power_series}Let a perfect fluid with an equation
of state~(\ref{General_Power_Series_eq:EoS_general}). If real analytic
solutions at $r=0$ of the TOV equations together with an equation
of state exist for compact stellar objects, then the energy density
is an even function of the circumferential radius coordinate, such
that its power series about $r=0$ contains only even powers of $r$.
\end{thm}
The proof of Theorem~\ref{Theorem:Even_power_series} is given in
Appendix~\ref{Section:Appendix_Proof-of-Theorem}. In proving this
theorem, we have also shown the following:
\begin{cor}
In the conditions of Theorem~\ref{Theorem:Even_power_series}, the
pressure, $p$, is an even function of the circumferential radius,
and $\mathcal{A}$ is an odd function.
\end{cor}
Theorem~\ref{Theorem:Even_power_series} and its corollary establish
a relation between the behavior of the thermodynamic variables in
a neighborhood of the center of the star and the non-vanishing of
the adiabatic speed of sound. These results allow us to simplify the
computation of the coefficients of the power series. Under the conditions
of Theorem~\ref{Theorem:Even_power_series}, the odd-order terms
of the power series expansion of $\mathds{P}$, Eq.~(\ref{Reformulation_eq:Power_series_P_matrix}),
are identically zero. Then, the recurrence relation~(\ref{Reformulation_eq:Recurrence_relation_P})
simplifies to
\begin{equation}
\mathds{P}_{m}=\begin{cases}
\mathds{I}_{2} & \text{, if }m=0\\
\boldsymbol{0} & \text{, if }m\text{ odd}\\
{\displaystyle \frac{1}{m}\sum^{\frac{m}{2}-1}_{s=0}\tilde{\Theta}_{m-1-2s}\,\mathds{P}_{2s}} & \text{, if }m\text{ even}
\end{cases}\label{General_Power_Series_eq:Recurrence_relation_P_simplified}
\end{equation}

Knowledge of the derivatives of an analytic function at a given point
can be used to find the power series representation of the function.
However, we can approximate the function using other type of series.
As we have discussed in the previous section, for instance, we can
use Pad\'e approximants. Using Eqs.~(\ref{General_Power_Series_eq:General_series}),
we present the general Pad\'e approximants of order $\left[2,2\right]$
for the energy density, $\mu_{\left[2,2\right]}$, and the pressure,
$p_{\left[2,2\right]}$, for an EoS characterized by Eq.~(\ref{General_Power_Series_eq:Derivative_EoS}):
\begin{align}
\mu_{\left[2,2\right]} & =\frac{\mu_{c}+\mathfrak{m}_{1}\,r^{2}}{1+\mathfrak{m}_{2}\,r^{2}}\,,\label{General_Power_Series_eq:Pade_mu_general}\\
p_{\left[2,2\right]} & =\frac{p_{c}-\frac{1}{60}\left[5\mu_{c}\left(\mu_{c}+4p_{c}\right)+\frac{p_{c}\left(4\mu_{c}+9p_{c}\right)f_{p}}{f_{\mu}}\right]r^{2}}{1+\left(\frac{p_{c}}{4}-\frac{\left(4\mu_{c}+9p_{c}\right)f_{p}}{60f_{\mu}}\right)r^{2}}\,,\label{General_Power_Series_eq:Pade_p_general}
\end{align}
where
\begin{equation}
\begin{aligned}\mathfrak{m}_{1} & =\frac{1}{4}\mu_{c}p_{c}+\frac{\left[5\mu_{c}\left(\mu_{c}+p_{c}\right)\left(\mu_{c}+3p_{c}\right)f_{\mu\mu}+2\left(\mu^{2}_{c}+15p^{2}_{c}+11\mu_{c}p_{c}\right)f_{\mu}\right]f_{p}}{120\left(f_{\mu}\right)^{2}}\\
 & +\frac{\mu_{c}\left(\mu_{c}+p_{c}\right)\left(\mu_{c}+3p_{c}\right)f_{pp}}{24f_{p}}-\frac{\mu_{c}\left(\mu_{c}+p_{c}\right)\left(\mu_{c}+3p_{c}\right)f_{\mu p}}{12f_{\mu}}\,,\\
\mathfrak{m}_{2} & =\frac{1}{4}p_{c}+\frac{\left[5\left(\mu_{c}+p_{c}\right)\left(\mu_{c}+3p_{c}\right)f_{\mu\mu}-2\left(4\mu_{c}+9p_{c}\right)f_{\mu}\right]f_{p}}{120\left(f_{\mu}\right)^{2}}\\
 & +\frac{\left(\mu_{c}+p_{c}\right)\left(\mu_{c}+3p_{c}\right)f_{pp}}{24f_{p}}-\frac{\left(\mu_{c}+p_{c}\right)\left(\mu_{c}+3p_{c}\right)f_{\mu p}}{12f_{\mu}}\,,
\end{aligned}
\end{equation}
and the derivatives of the function $f$ are evaluated at $\left(\mu_{c},p_{c}\right)$.
We remark that in the light of Theorem~\ref{Theorem:Even_power_series},
the Pad\'e approximants for the energy density and the pressure are
even functions. In particular, for each of these functions, the Pad\'e
approximant of order $\left[2n,2n\right]$ is equal to the Pad\'e
approximant of order $\left[2n+1,2n+1\right]$, for $n\geq0$, if
they exist.

In Section~\ref{subsec:Examples_given_mu}, we have introduced the
idea of finding an analytic approximation for the radius of the star,
$r_{b}$, defined by Eq.~(\ref{Reformulation_eq:radius_star_definition}),
by considering the Pad\'e approximant of order $\left[2,2\right]$
of the pressure. Evaluating the roots of $p_{\left[2,2\right]}$,
Eq.~(\ref{General_Power_Series_eq:Pade_p_general}), yields the following
approximation in terms of the EoS,
\begin{equation}
r_{b}\approx\sqrt{\frac{60p_{c}}{5\mu_{c}\left(\mu_{c}+4p_{c}\right)+\frac{p_{c}\left(4\mu_{c}+9p_{c}\right)f_{p}}{f_{\mu}}}}\,.\label{General_Power_Series_eq:Star_radius_general_approximation}
\end{equation}
Following the same reasoning, we can use Eq.~(\ref{General_Power_Series_eq:General_series})
in Eq.~(\ref{Reformulation_eq:Mass_function_definition}) and compute
the Pad\'e approximant of order {[}3,4{]} to find an approximation
for the value of the mass function at a given circumferential radius.
We find
\begin{equation}
M\approx\frac{\mu_{c}r^{3}}{6-\frac{3\left(\mu_{c}+p_{c}\right)\left(\mu_{c}+3p_{c}\right)f_{p}}{10\mu_{c}f_{\mu}}r^{2}+\frac{\left(\mu_{c}+p_{c}\right)\left(\mu_{c}+3p_{c}\right)\mathsf{M}}{2800\mu^{2}_{c}\left(f_{\mu}\right)^{3}}r^{4}}\,,\label{General_Power_Series_eq:Mass_function_general_approximation}
\end{equation}
where
\begin{equation}
\begin{aligned}\mathsf{M} & =\left[2\left(\mu^{2}_{c}+63p^{2}_{c}+39\mu_{c}p_{c}\right)f_{\mu}+25\mu_{c}\left(\mu_{c}+p_{c}\right)\left(\mu_{c}+3p_{c}\right)f_{\mu\mu}\right]\left(f_{p}\right)^{2}\\
 & -50\mu_{c}f_{\mu}\left[\left(\mu_{c}+p_{c}\right)\left(\mu_{c}+3p_{c}\right)f_{\mu p}-3p_{c}f_{\mu}\right]f_{p}\\
 & +25\mu_{c}\left(\mu_{c}+p_{c}\right)\left(\mu_{c}+3p_{c}\right)f_{pp}\left(f_{\mu}\right)^{2}\,.
\end{aligned}
\end{equation}
Equations.~(\ref{General_Power_Series_eq:Star_radius_general_approximation})
and (\ref{General_Power_Series_eq:Mass_function_general_approximation})
can be used to estimate the compactness parameter: $M\left(r_{b}\right)/r_{b}$,
from the values of $\mu_{c}$, $p_{c}$ and the derivatives up to
second order of the function $f$.

To close this section, we note that it is straightforward to find
symbolic expressions for analytic approximations for $r_{b}$ and
$M$, using higher order approximants. However, the expressions for
the coefficients of the polynomials become increasingly larger, so
that it is not sensible to present general formulas for higher order
analytic approximations. In the next section, we will consider the
higher order approximations for particular EoS.

\section{Applications to selected equations of state\label{Section:Applications_and_examples}}

\subsection{Affine equation of state}

Affine EoS are simple, important models used to characterize states
of matter in the interior of stars~\citep{Witten_1984,Christodoulou_1995,Urbano_Veermae_2019,Sharma_Maharaj_2007}.
Generically, affine EoS can be written as
\begin{equation}
p-\alpha\mu-\beta=0\,,
\end{equation}
where $\alpha$ and $\beta$ are constants. For an affine EoS, the
function $f$ is characterized by $f_{\mu}=-\alpha$ and $f_{p}=1$,
and all higher order derivatives are identically zero.

To exemplify the general analytic results derived in Section~\ref{Section:General_Solution_EOS},
we will consider the particular realization of an affine EoS provided
by the MIT bag model. The MIT bag model follows from the assumption
that hadrons are composed of non-interacting, massless quarks, behaving
as a Fermi gas with energy density $\mu$ and pressure $p$, confined
in a region by an external, counterbalancing pressure, $B$, that
maintains the quark gas at finite density and chemical potential.
In this model, the fluid's EoS is
\begin{equation}
\mu=3p+4B\,.\label{Examples_EoS_eq:MIT_EoS}
\end{equation}

An application of the MIT EoS to compact stellar objects is to characterize
the matter fluid source of Strange Stars. In Ref.~\citep{Farhi_Jaffe_1984},
it was shown that three-flavor, $\left(u,d,s\right)$, quark matter
is stable for values of $B$ between $57\,\text{MeV/fm}^{3}$ and
$94\,\text{MeV/fm}^{3}$, allowing for the possibility that stellar
objects composed of strange quark matter might exist in the universe.
Such objects are often referred as Strange stars. We will then analyze
the properties of these objects by numerically integrating the TOV
equations for a fluid characterized by the MIT EoS~(\ref{Examples_EoS_eq:MIT_EoS}),
and compare the results with those of the analytic approximations
introduced in the previous section. To allow easy comparison with
the results in the literature on Strange Stars, we will present the
results in terms of the bag constant, $B$, in $\text{MeV/fm}^{3}$.

Figure~\ref{Figure:MIT_model_radius_strange_stars} shows the behavior
of the radius of Strange stars as a function of the bag constant,
following from numerical integration of the TOV equations together
with the EoS~(\ref{Examples_EoS_eq:MIT_EoS}), and the analytic approximation~(\ref{General_Power_Series_eq:Star_radius_general_approximation}),
and the respective percent error. Strange stars solutions aim to represent
highly compact configurations, where the matter fields are subjected
to extreme forces. As expected, in these scenarios, the analytic approximation
to the radius fails to fully capture the complexities of the behavior
of the fluid's pressure: the approximations becomes increasingly worse,
as the value of the bag constant decreases.

The radius estimate is important, since it used in the analytic estimation
of the total mass of the star, $M\left(r_{b}\right)$, using Eq.~(\ref{General_Power_Series_eq:Mass_function_general_approximation}).
Given the simplicity of the EoS~(\ref{Examples_EoS_eq:MIT_EoS}),
in this case we can write the expression for the analytic approximation
for $r_{b}$ using the diagonal Pad\'e approximant of order $\left[4,4\right]$
of the pressure in a somewhat compact form. We find
\begin{equation}
r_{b}\approx\sqrt{\frac{\mathcal{R}_{1}}{2\mathcal{R}_{2}}+\varepsilon\frac{\sqrt{\mathcal{R}^{2}_{1}+4\mathcal{R}_{2}\,p_{c}}}{2\mathcal{R}_{2}}}\,,\label{Examples_EoS_eq:Star_radius_second order_approximation}
\end{equation}
where $\varepsilon=\pm1$ and
\begin{equation}
\begin{aligned}\mathcal{R}_{1} & =\frac{-134\mu^{4}_{c}+9513p^{4}_{c}+6336\mu_{c}p^{3}_{c}-1321\mu^{2}_{c}p^{2}_{c}-1078\mu^{3}_{c}p_{c}}{24\left(67\mu^{2}_{c}+567p^{2}_{c}+459\mu_{c}p_{c}\right)}\,,\\
\mathcal{R}_{2} & =\frac{1904\mu^{5}_{c}-1\,154\,979p^{5}_{c}-1\,246\,518\mu_{c}p^{4}_{c}-266\,706\mu^{2}_{c}p^{3}_{c}+34\,330\mu^{3}_{c}p^{2}_{c}+15401\mu^{4}_{c}p_{c}}{10\,080\left(67\mu^{2}_{c}+567p^{2}_{c}+459\mu_{c}p_{c}\right)}\,.
\end{aligned}
\end{equation}
The sign of $\varepsilon$ is chosen so that the left-hand side of
Eq.~(\ref{Examples_EoS_eq:Star_radius_second order_approximation})
is real positive. If the approximation for $r_{b}$ is real positive
for both signs of $\varepsilon$, the one yielding the smallest value
is chosen.

As illustrated in Figure~\ref{Figure:MIT_model_radius_strange_stars},
the increased accuracy of the higher order analytic approximation~(\ref{Examples_EoS_eq:Star_radius_second order_approximation})
in comparison with that of Eq.~(\ref{General_Power_Series_eq:Star_radius_general_approximation})
is significant for all values of the bag constant. The fact that the
analytic approximation~(\ref{Examples_EoS_eq:Star_radius_second order_approximation})
is accurate for such extreme configurations is a consequence of the
simplicity of the EoS.

\begin{figure}
\includegraphics[totalheight=0.19\paperheight]{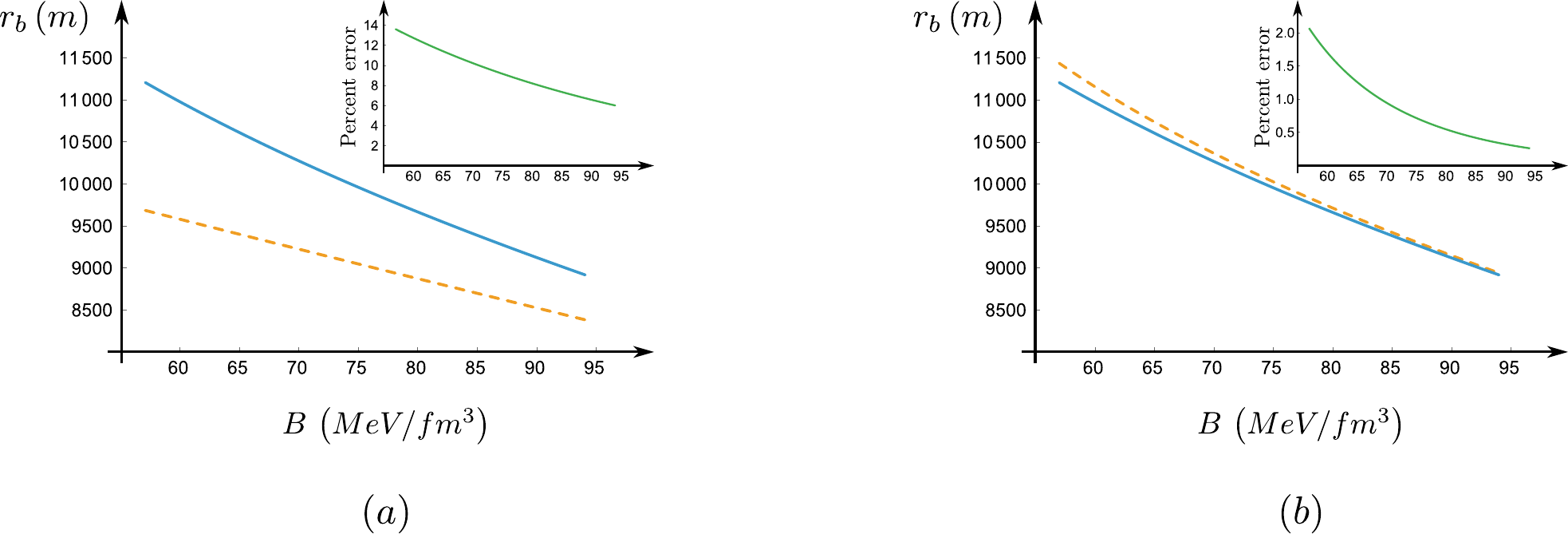}

\caption{\label{Figure:MIT_model_radius_strange_stars}Radius of Strange stars
as a function of the bag constant, assuming central energy density
$\mu_{c}=900\,\text{MeV/fm}^{3}$ In both plots the solid blue line
follows from numerical integration of the TOV equations. In subfigure
($a$), the dashed orange line follows from the analytic approximation~(\ref{General_Power_Series_eq:Star_radius_general_approximation}).
In subfigure ($b$), the dashed orange line follows from the analytic
approximation~(\ref{Examples_EoS_eq:Star_radius_second order_approximation}).
The inset panels in each subfigure show the percent error of the respective
analytic approximations relative to the numerical result.}
\end{figure}

In Figure~\ref{Figure:MIT_model_Compactness_strange_stars}, we present
the compactness parameter, $M/r_{b}$, of strange stars as a function
of the bag constant following from numerical integrating the TOV equations,
and the analytic approximation following from Eqs.~(\ref{General_Power_Series_eq:Mass_function_general_approximation})
and (\ref{Examples_EoS_eq:Star_radius_second order_approximation}).
Expectedly, the errors associated with using the analytic expressions
become worse for higher values of the compactness parameter, where
the behavior of the energy density and the pressure changes significantly,
particularly, closer to the boundary of the star, and higher order
approximations are required to better fit the model. Nonetheless,
the percent error of the analytic approximation relative to the numerical
result is surprisingly low, ranging from 1.4\% to 4.1\%.

\begin{figure}
\includegraphics[totalheight=0.16\paperheight]{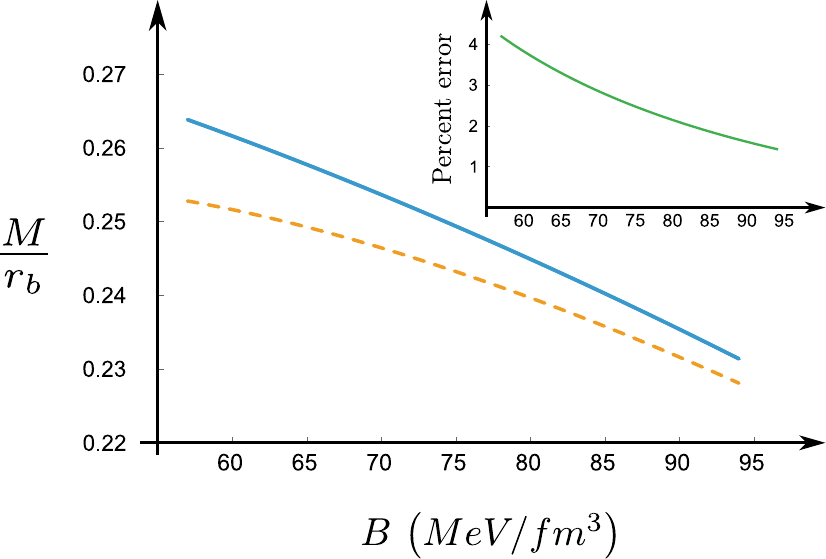}\caption{\label{Figure:MIT_model_Compactness_strange_stars}Compactness parameter
of Strange stars as a function of the bag constant, assuming central
energy density $\mu_{c}=900\,\text{MeV/fm}^{3}$. Solid line follows
from numerical integration of the TOV equations. Dashed curve follows
from using the analytic approximations~(\ref{General_Power_Series_eq:Mass_function_general_approximation})
and (\ref{Examples_EoS_eq:Star_radius_second order_approximation}).
The inset panel shows the percent error of the analytic approximation
relative to the numerical result.}
\end{figure}

\subsection{Relativistic polytropic equation of state}

Relativistic polytropic EoS are a popular family of equations of state
to model matter in the interior of compact stellar objects. In these
models,
\begin{equation}
p=K\,\rho^{\gamma}\,,
\end{equation}
where $\rho$ represents the matter rest mass density, $\gamma$ is
the adiabatic index defined in Eq.~(\ref{General_Power_Series_eq:Adiabatic_index}),
and the polytropic constant $K$ is determined by the thermodynamic
properties of the perfect fluid at a particular point. For instance,
considering the values of the rest mass density and the pressure at
the center of the star, $\rho_{c}$ and $p_{c}$, respectively, we
have
\begin{equation}
K=\frac{p_{c}}{\rho^{\gamma}_{c}}\,.
\end{equation}
For a relativistic polytropic EoS, Eq.~(\ref{General_Power_Series_eq:EoS_general})
reads
\begin{equation}
\left(\frac{p}{K}\right)^{\frac{1}{\gamma}}+\frac{p}{\gamma-1}-\mu=0\,.\label{Examples_EoS_eq:Relativistic_polytropic_EoS}
\end{equation}

The value of the adiabatic index, $\gamma$, and the polytropic constant,
$K$, markedly change the properties of the fluid, such that relativistic
polytropic EoS have been considered to model distinct types of compact
stellar objects. Indeed, fluids characterized by relativistic polytropic
EoS have been extensively considered as effective models, instead
of microphysics based models of the equation of state (see, e.g.,
Refs.\citep{Tooper_1965,Nilsson_Uggla_2001,Heinzle_Rohr_Uggla_2003,Herrera_Barreto_2013,Herrera_et_al_2014}
and references therein).

In general, the thermodynamic properties of fluids described by a
relativistic polytropic EoS generaly have more complex behavior than
those of fluids characterized by an affine equation of state. Depending
on the values of the parameters, the analytic approximations discussed
previously might give very poor results or fail completely. For instance,
if the pressure distribution is characterized by a long decreasing
tail, the analytic approximation for the radius of the star~(\ref{Examples_EoS_eq:Star_radius_second order_approximation})
might yield imaginary results. Nonetheless, we have found that the
Pad\'e approximants converge fast to the solutions, yielding accurate
results by including a few more terms of the series.

To exemplify the accuracy of analytic series solutions, in Table~\ref{Table:Polytropic_comparison}
we present the values of the radius and the compactness parameter
of a self-gravitating fluid characterized by a relativistic polytropic
EoS for different values of the adiabatic index, for fixed values
of the central density and pressure. For comparison, we show the results
for the pressure and the mass function from integrating the TOV equations
numerically and the analytic results using the diagonal Pad\'e approximant
of the order $\left[10,10\right]$.

Depending on the values of $\gamma$ and $K$, numerical integration
of the TOV equations can be challenging. In that regard, we have found
that for this type of EoS, it is preferable to use the formulation
of the TOV equations introduced by Lindblom~\citep{Lindblom_1998}
to solve the equations numerically.

As we see, for the considered values, the analytic results agree with
the numerical approximations, with small percent errors. For smaller
values of $\gamma$, the pressure profile has longer decreasing tails
toward the boundary, such that the percent error increases with decreasing
values of $\gamma$.

\begin{table}
\begin{tabular}{|c|c|c|c||c|c|c|}
\hline 
 & {\small ~~$r_{b}\,\left(\text{Num.}\right)$~~} & {\small ~~$r_{b}\,\left(\text{An. approx.}\right)$~~} & ~~~~~~$\delta$~~~~~~ & {\small ~~$\frac{M}{r_{b}}\,\left(\text{Num.}\right)$~~} & {\small ~~$\frac{M}{r_{b}}\,\left(\text{An. approx.}\right)$~~} & ~~~~~~$\delta$~~~~~~\tabularnewline
\hline 
\hline 
{\small ~~$\gamma=2.0$~~} & {\small 20~624} & {\small 19~584} & {\small 5.0\%} & {\small 0.190} & {\small 0.198} & {\small 4.4\%}\tabularnewline
\hline 
{\small$\gamma=2.2$} & {\small 18~874} & {\small 18~411} & {\small 2.5\%} & {\small 0.200} & {\small 0.204} & {\small 2.3\%}\tabularnewline
\hline 
{\small$\gamma=2.4$} & {\small 17~748} & {\small 17~445} & {\small 1.7\%} & {\small 0.206} & {\small 0.210} & {\small 1.6\%}\tabularnewline
\hline 
{\small$\gamma=2.6$} & {\small 16~957} & {\small 16~729} & {\small 1.3\%} & {\small 0.211} & {\small 0.213} & {\small 1.3\%}\tabularnewline
\hline 
{\small$\gamma=2.8$} & {\small 16~370} & {\small 16~171} & {\small 1.2\%} & {\small 0.214} & {\small 0.216} & {\small 1.1\%}\tabularnewline
\hline 
{\small$\gamma=3.0$} & {\small 15~915} & {\small 15~742} & {\small 1.1\%} & {\small 0.217} & {\small 0.219} & {\small 1.0\%}\tabularnewline
\hline 
\end{tabular}

\caption{\label{Table:Polytropic_comparison} Approximate values of the radius,
in meters, and the compactness parameter found from numerical integration
of the TOV equations with a relativistic EoS~(\ref{Examples_EoS_eq:Relativistic_polytropic_EoS}),
and computed considering the analytic Pad\'e approximant of order
$\left[10,10\right]$, for different values of the adiabatic index
$\gamma$, with central density $\rho_{c}=5\times10^{14}\,\text{g/cm}^{3}$
and central pressure $p_{c}=10^{34}\,\text{Pa}$. The fourth and last
columns list the approximate values of the percent error, $\delta$,
of the analytic approximation relative to the numerical approximation
for the radius and the compactness, respectively.}
\end{table}

\section{Solutions for piecewise equations of state\label{Section:Piecewise}}

\subsection{Series solutions about arbitrary points\label{subsection:Piecewise_general_series}}

Properties of the matter fields may vary significantly as we move
from the inner stellar core toward the outer crust regions. Determining
from microphysical models an EoS that characterizes fluid sources
across all regimes within compact object interiors is extremely challenging.
Moreover, stellar objects may be structured in stratified layers,
where matter exists in distinct states resulting from phase transitions,
such that thermodynamic variables may be discontinuous across transition
layers. It is therefore common to adopt piecewise EoS models, in which
each branch describes the matter fluid in a particular regime~\citep{Lindblom_2010,Read_et_al_2009,Boyle_2020,Suleiman_et_al_2022,Lindblom_1992,Vuille_2008}.

In this section, we derive solutions to the TOV-EoS system, for piecewise
equations of state. 

In Eq.~(\ref{Reformulation_eq:System_A_equation}), we have separated
the singular part from the non-singular part to apply the Coddington-Levison
algorithm to find series solutions about $r=0$. Assuming $\Theta$
is regular, the system does not have other singular points. Then,
if we are interested in determining series solutions in a neighborhood
centered at any point other than at $r=0$, we may write the differential
system simply as
\begin{equation}
\frac{d\mathds{W}}{dr}=\mathds{A}\left(r\right)\mathds{W}\,,
\end{equation}
where
\begin{equation}
\mathds{A}\left(t\right)=\left[\begin{array}{cc}
\frac{1}{r}+\frac{4\mu+6\mathcal{E}}{3r\phi^{2}} & \frac{6\mathcal{E}}{r^{2}\phi^{2}}\\
1 & 0
\end{array}\right]\,.
\end{equation}
The general analytic solutions around a given point, say $r=r_{J}>0$,
are formally given by
\begin{equation}
\mathds{W}=\mathds{C}\left(r\right)\left[\begin{aligned}c_{1}\\
c_{2}
\end{aligned}
\right]\,,\label{Piecewise_eq:Solution_s_u}
\end{equation}
where $c_{1}$ and $c_{2}$ are integration constants, and the matrix
$\mathds{C}$ has power series
\begin{equation}
\mathds{C}\left(r\right)=\sum^{+\infty}_{m=0}\mathds{C}_{m}\left(r-r_{J}\right)^{m}\,,
\end{equation}
where the matrix coefficients $\mathds{C}_{m}$ are given by the recurrence
relation
\begin{equation}
\begin{aligned}\mathds{C}_{0} & =\mathds{I}_{2}\,,\\
\mathds{C}_{m} & =\frac{1}{m}\sum^{m-1}_{k=0}\mathds{A}_{m-1-k}\mathds{C}_{k}\,,\text{for }\ensuremath{m\geq}1\,.
\end{aligned}
\label{Piecewise_eq:Solution_C_recurrence_relation_2}
\end{equation}

Equations~(\ref{Piecewise_eq:Solution_s_u})--(\ref{Piecewise_eq:Solution_C_recurrence_relation_2})
define the power series solutions to the TOV equations about a particular
hypersurface. These can be used to construct solutions, such that
the spacetime is characterized by piecewise potentials. Let $r=r_{J}$
define the junction hypersurface within the star. From Eqs.~(\ref{Reformulation_eq:Solution_s_u_A})
and (\ref{Reformulation_eq:Solution_s_u_p}), setting, for instance,
$c_{2}=1$, we find that $c_{1}$ is given in terms of the energy
density and the pressure at the junction, $\mu_{J}$ and $p_{J}$,
respectively, as
\begin{equation}
c_{1}=\frac{r_{J}}{2\left(1-\frac{2M_{J}}{r_{J}}\right)}\left(p_{J}+\frac{2M_{J}}{r^{3}_{J}}\right)\,,\label{Piecewise_eq:Integration_constant_c1}
\end{equation}
where
\begin{equation}
M_{J}=\frac{1}{2}\int^{r_{J}}_{0}\mu\left(x\right)\,x^{2}\,dx\,,\label{Piecewise_eq:Enclosed_mass_transition}
\end{equation}
 represents the value of the mass at the junction hypersurface.

Similarly to the algorithm introduced in Section~\ref{Section:General_Solution_EOS},
using Eqs.~(\ref{Reformulation_eq:Solution_s_u_A}), (\ref{Reformulation_eq:Solution_s_u_p})
and (\ref{General_Power_Series_eq:EoS_derivative_mu}), and Eqs.~(\ref{Piecewise_eq:Solution_s_u})--(\ref{Piecewise_eq:Integration_constant_c1})
we can compute the power series coefficients of $\mu$, $p$ and $\mathcal{A}$
at $r_{J}$, for a given EoS of the form~(\ref{General_Power_Series_eq:EoS_general}).
We present below the first terms of the power series:
\begin{equation}
\begin{aligned}\mu\left(r\right) & =\mu_{J}+\frac{\left(\mu_{J}+p_{J}\right)\left(2M_{J}+p_{J}r^{3}_{J}\right)f_{p}}{2r_{J}\left(r_{J}-2M_{J}\right)f_{\mu}}\left(r-r_{J}\right)\\
 & -\frac{\left(r_{J}-2M_{J}\right)\left(\mu_{J}+p_{J}\right)\mathfrak{m}_{J}}{8r^{2}_{J}\left[\left(2M_{J}-r_{J}\right)f_{\mu}\right]^{3}}\left(r-r_{J}\right)^{2}+\mathcal{O}\left(\left(r-r_{J}\right)^{3}\right)\,,\\
p\left(r\right) & =p_{J}-\frac{\left(\mu_{J}+p_{J}\right)\left(2M_{J}+p_{J}r^{3}_{J}\right)}{2\left[r_{J}\left(r_{J}-2M_{J}\right)\right]}\left(r-r_{J}\right)\\
 & -\frac{\left(\mu_{J}+p_{J}\right)\mathfrak{p}_{J}}{8r^{2}_{J}\left(r_{J}-2M_{J}\right)^{2}f_{\mu}}\left(r-r_{J}\right)^{2}+\mathcal{O}\left(\left(r-r_{J}\right)^{3}\right)\,,\\
\mathcal{A}\left(r\right) & =\frac{2M_{J}+p_{J}r^{3}_{J}}{2r^{2}_{J}\sqrt{1-\frac{2M_{J}}{r_{J}}}}+\frac{1}{4r^{4}_{J}\left(1-\frac{2M_{J}}{r_{J}}\right)^{\frac{3}{2}}}\left[12M^{2}_{J}-4M_{J}r_{J}\left(2p_{J}r^{2}_{J}+\mu_{J}r^{2}_{J}+2\right)\right.\\
 & \left.+r^{4}_{J}\left(2\mu_{J}-p^{2}_{J}r^{2}_{J}+2p_{J}\right)\right]\left(r-r_{J}\right)+\mathcal{O}\left(\left(r-r_{J}\right)^{2}\right)\,,
\end{aligned}
\label{Piecewise_eq:Series_solutions_mu_p_A_general}
\end{equation}
where the derivatives of $f$ are taken at $\left(\mu_{J},p_{J}\right)$
and
\begin{equation}
\begin{aligned}\mathfrak{m}_{J} & =f_{p}\left(f_{\mu}\right)^{2}\left[4M^{2}_{J}-2M_{J}r_{J}\left(7p_{J}r^{2}_{J}+\mu_{J}r^{2}_{J}+4\right)+r^{4}_{J}\left(2\mu_{J}-2p^{2}_{J}r^{2}_{J}+\mu_{J}p_{J}r^{2}_{J}+2p_{J}\right)\right]\\
 & -\left(\mu_{J}+p_{J}\right)\left(2M_{J}+p_{J}r^{3}_{J}\right)^{2}f_{\mu}\left(f_{pp}f_{\mu}-f_{p}f_{\mu p}\right)+\left(2M_{J}+p_{J}r^{3}_{J}\right)^{2}f_{\mu}\left(f_{p}\right)^{2}\\
 & -\left(\mu_{J}+p_{J}\right)\left(2M_{J}+p_{J}r^{3}_{J}\right)^{2}f_{p}\left(f_{p}f_{\mu\mu}-f_{\mu}f_{\mu p}\right)\,,\\
\mathfrak{p}_{J} & =f_{\mu}\left(4M^{2}_{J}-2M_{J}r_{J}\left(7p_{J}r^{2}_{J}+\mu_{J}r^{2}_{J}+4\right)+r^{4}_{J}\left(2\mu_{J}-2p^{2}_{J}r^{2}_{J}+\mu_{J}p_{J}r^{2}_{J}+2p_{J}\right)\right)\\
 & +\left(2M_{J}+p_{J}r^{3}_{J}\right)^{2}f_{p}\,.
\end{aligned}
\label{Piecewise_eq:Series_solutions_auxiliar_quantities}
\end{equation}
Note that, since $\mathcal{A}\left(r_{J}\right)$ is not necessarily
zero, the power series solutions for $\mu$ and $p$ may contain terms
with odd-powers of $\left(r-r_{J}\right)$. Furthermore, these power
series can be used to construct other series solutions about $r_{J}$,
such as Pad\'e approximants.

Equations~(\ref{Piecewise_eq:Series_solutions_mu_p_A_general}) and
(\ref{Piecewise_eq:Series_solutions_auxiliar_quantities}) depend
on derivatives of the equation of state with respect to the energy
density and the pressure at the junction points: $\left(\mu_{J},p_{J}\right)=\left(\mu\left(r_{J}\right),p\left(r_{J}\right)\right)$.
For piecewise EoS, some derivatives might not exist. Nonetheless,
the series solutions from Eqs.~(\ref{Piecewise_eq:Solution_s_u})--(\ref{Piecewise_eq:Integration_constant_c1})
aim to describe a region where $r>r_{J}$. Therefore, only the one-sided
derivatives of $f$, in the direction $r\to r^{+}_{J}$, are required.
As such, the EoS does not have to be smooth, or even continuous across
the matching hypersurfaces.

To close this subsection, we remark that the Israel-Darmois junction
conditions must be verified at the matching hypersurfaces for the
full interior spacetime to be a solution of the Einstein field equations~\citep{Darmois_1927,Israel_1966}.
Whether the junction is smooth or requires the presence of thin matter
shells at the matching hypersurfaces, determines the values of $p_{J}$
and $M_{J}$. Nonetheless, matching of the series solutions is well
defined. 

\subsection{Example: Piecewise polytrope}

To exemplify the analytic results of the previous subsection and see
how they compare with those found from numerical integration, we will
consider a simple toy model of an EoS defined as a piecewise polytrope
characterizing three regions within a stellar object. The three subregions
are determined by the density of the fluid:
\begin{equation}
p\left(\rho\right)=\begin{cases}
K_{1}\,\rho^{\gamma_{1}} & \text{, }\rho_{c}\leq\rho<\rho_{1}\\
K_{2}\,\rho^{\gamma_{2}} & \text{, }\rho_{1}\leq\rho<\rho_{2}\\
K_{3}\,\rho^{\gamma_{3}} & \text{, }\rho\geq\rho_{2}
\end{cases}
\end{equation}
where $\rho$ represents the matter rest mass density and $\gamma_{i}$,
with $i\in\left\{ 1,2,3\right\} $, are the adiabatic indexes for
each region. We assume the values of the adiabatic index for each
subdomain are known values. To determine the values of the polytropic
constants, we consider that the pressure is continuous throughout
the star, but not differentiable. Imposing these conditions, we find~
\begin{equation}
p\left(\rho\right)=\begin{cases}
{\displaystyle \frac{p_{c}}{\rho^{\gamma_{1}}_{c}}\,\rho^{\gamma_{1}}} & \text{, }\rho_{c}\leq\rho<\rho_{1}\\
{\displaystyle \frac{\rho^{\gamma_{1}-\gamma_{2}}_{1}p_{c}}{\rho^{\gamma_{1}}_{c}}\,\rho^{\gamma_{2}}} & \text{, }\rho_{1}\leq\rho<\rho_{2}\\
{\displaystyle \frac{\rho^{\gamma_{1}-\gamma_{2}}_{1}\rho^{\gamma_{2}-\gamma_{3}}_{2}p_{c}}{\rho^{\gamma_{1}}_{c}}\,\rho^{\gamma_{3}}} & \text{, }\rho\geq\rho_{2}
\end{cases}\label{Piecewise_eq:Piecewise_politropic_EoS}
\end{equation}

In Table~\ref{Table:Piecewise_Polytropic_comparison}, we present
a comparison between values found from numerically integrating the
TOV equations and from the analytic series solutions, for particular
values of the model parameters. Namely, we show approximate values
of the total masses within the transition hypersurfaces and the respective
radii, and the total mass and radius of the star. 

The analytic series solution is defined piecewise. For the inner core
region, where $\rho<\rho_{1}$, we considered the Pad\'e approximant
about $r=0$ of order $\left[10,10\right]$, following from Eqs.~(\ref{Reformulation_eq:Solution_s_u}),
both for the pressure and the energy density. For the regions where
$\rho_{1}\leq\rho\leq\rho_{2}$ and $\rho\geq\rho_{2}$, we considered
the Pad\'e approximants of order $\left[3,3\right]$, respectively,
about the transition radii $\left(r_{J}\right)_{1}$ and $\left(r_{J}\right)_{2}$,
following from Eqs.~(\ref{Piecewise_eq:Solution_s_u})--(\ref{Piecewise_eq:Enclosed_mass_transition}).

Comparing the results, we see that the numerical integration and the
analytic approximations are in close agreement. We remark that, while
analytic expressions can be evaluated rapidly, numerical integration
can be particularly slow for certain values of the adiabatic index.

\begin{table}
\begin{tabular}{|c|c|c|c|c|c|c|}
\hline 
 & {\small ~~$\left(r_{J}\right)_{1}$~~} & {\small ~~$\frac{\left(M_{J}\right)_{1}}{M_{\odot}}$~~} & {\small ~~$\left(r_{J}\right)_{2}$} ~~ & {\small ~~$\frac{\left(M_{J}\right)_{2}}{M_{\odot}}$~~} & {\small ~~~~$r_{b}$~~~~} & {\small ~~~$\frac{M}{M_{\odot}}$~~~}\tabularnewline
\hline 
\hline 
{\small ~~Numeric~~} & {\small 13~929} & {\small 2.043} & {\small 16~657} & {\small 2.341} & {\small 16~701} & {\small 2.343}\tabularnewline
\hline 
{\small ~~Analytic~~} & {\small 13~929} & {\small 1.974} & {\small 16~681} & {\small 2.344} & {\small 16~682} & {\small 2.344}\tabularnewline
\hline 
\end{tabular}

\caption{\label{Table:Piecewise_Polytropic_comparison} Approximate values
of the radii, in meters, the masses within the transition hypersurfaces,
and the total mass and radius of a star characterized by a piecewise
EoS~(\ref{Piecewise_eq:Piecewise_politropic_EoS}), with central
density $\rho_{c}=5\times10^{14}\,\text{g/cm}^{3}$, central pressure
$p_{c}=10^{34}\,\text{Pa}$, transition rest mass densities: $\rho_{1}=2\times10^{14}\,\text{g/cm}^{3}$
and $\rho_{2}=10^{12}\,\text{g/cm}^{3}$, and the polytropic indexes
$\left(\gamma_{1},\gamma_{2},\gamma_{3}\right)=\left(3,2,\frac{4}{3}\right)$.
The second row shows the values obtained from numerical integration.
The third row shows the values obtained from the analytic series solutions.}
\end{table}

\section{Conclusion\label{Section:Conclusions}}

Assuming sufficient regularity of the matter fluid, we have shown
that the TOV equations can be solved in terms of analytic series.
The series solutions developed in this work allowed us to obtain closed-form
approximations to macroscopic properties of compact stellar objects,
either in terms of the spacetime parameters or in terms of the underlying
equation of state of the fluid.

The considered examples demonstrate that analytic solutions exist
for widely used, nontrivial EoS, and general approximation formulas
for the stellar radius and mass can be derived. The accuracy of the
analytic approximations depends on the behavior of the thermodynamic
variables. In cases where the energy density or the pressure are characterized
by long decreasing tails, more terms of the series are required. Nevertheless,
using Pad\'e approximants, the series converge sufficiently fast,
such that closed-form expressions based on only a few terms of the
series yielded accurate results in the cases examined.

Realistic stellar objects are composed of matter in distinct physical
regimes. To model the stratification of stellar interiors, piecewise
equations of state are often adopted. In such cases, the thermodynamic
potentials may fail to be differentiable, or even continuous, across
junction hypersurfaces. To deal with such configurations, we have
constructed general series expansions about arbitrary points, such
that analytic solutions can be obtained for piecewise EoS and different
regions are characterized by different series, matched across the
transition hypersurfaces. To test this approach, we have considered
a piecewise relativistic polytropic EoS as a toy model. The analytic
results closely matched the numerical integration, having the advantage
of being much faster to evaluate.

For piecewise solutions, we have not derived approximate formulas
for the radius and mass functions. For piecewise EoS, analytic approximations
depend explicitly on the number of transition hypersurfaces. Nevertheless,
for a fixed number of transitions, analytic approximation formulas
can be constructed following the same reasoning as in the single-expression
EoS case.
\begin{acknowledgments}
The author thanks the Funda\c{c}\~{a}o para a Ci\^encia e Tecnologia
(FCT), Portugal, for the financial support to the Center for Astrophysics
and Gravitation through grant No. UID/PRR/00099/2025 and grant No.
UID/00099/2025.
\end{acknowledgments}

\appendix

\section{\label{Section:Appendix_Proof-of-Theorem}Proof of the parity property
of power series solutions}

Here we will prove Theorem~\ref{Theorem:Even_power_series}, stating
that if solutions to the TOV equations coupled with an equation of
state exist, such that $\mu$ is real analytic at $r=0$, its power
series contains only even powers of the circumferential radius.

We will prove a number of intermediate results which we summarize
in a set of lemmas. To avoid unnecessary repetition, the matrix $\widetilde{\Theta}$
is given by Eq.~(\ref{Reformulation_eq:Theta_tilde}) and has power
series expansion $\widetilde{\Theta}\left(r\right)=\sum^{+\infty}_{m=0}\widetilde{\Theta}_{m}\,r^{m}$;
it is assumed a barotropic EoS of the form~(\ref{General_Power_Series_eq:EoS_general}),
analytic at the center of the star and such that $f_{\mu}\left(\mu_{c},p_{c}\right)\neq0$,
where $\mu_{c}\equiv\mu\left(0\right)$ and $p_{c}\equiv p\left(0\right)$
represent the values of the central energy density and central pressure,
respectively, and $f_{\mu}\equiv\partial_{\mu}f$ and $f_{p}\equiv\partial_{p}f$.
\begin{lem}
\label{Lemma:Vanishin_odd_derivative_EoS}Let $f\left(\mu,p\right)=0$
represent the equation of state of the fluid source, defining $p$
implicitly as a function of $\mu$. For $n\geq1$ odd, if $\frac{f_{p}}{f_{\mu}}\in\mathcal{C}^{n}\left(\mu\left(0\right)\right)$
and $\mu\in\mathcal{C}^{n}\left(0\right)$, and for all odd $k\in\left\{ 1,...,n\right\} $,
$\mu^{\left(k\right)}$$\left(0\right)$ vanish, then
\begin{equation}
\frac{d^{n}}{dr^{n}}\left(\frac{f_{p}}{f_{\mu}}\right)\left(0\right)=0\,.\label{Appendix_eq:Faa_di_Bruno_ratio_EoS_at_zero}
\end{equation}
\end{lem}
\begin{proof}
A barotropic EoS $f\left(\mu,p\right)=0$ defines $p=p\left(\mu\right)$
implicitly as a function of $\mu$, such that $\frac{f_{p}}{f_{\mu}}$
is effectively a function of $\mu$. In turn $\mu$ is a function
of the radial coordinate $r$. If $\frac{f_{p}}{f_{\mu}}\in\mathcal{C}^{n}\left(\mu\left(0\right)\right)$
and $\mu\in\mathcal{C}^{n}\left(0\right)$, we can use directly Faà
di Bruno's formula for the $n^{th}$ derivative of the composition,
yielding
\begin{equation}
\frac{d^{n}}{dr^{n}}\left(\frac{f_{p}}{f_{\mu}}\right)\left(0\right)=\sum\frac{n!}{m_{1}!m_{2}!\,\cdots\,m_{n}!}\left(\frac{f_{p}}{f_{\mu}}\right)^{\left(m_{1}+m_{2}+\,\cdots\,+m_{n}\right)}\left(\mu\left(0\right)\right)\times\prod^{n}_{i=1}\left(\frac{1}{i!}\mu^{\left(i\right)}\left(0\right)\right)^{m_{i}}\,,\label{Appendix_eq:Faa_di_Bruno_ratio_EoS}
\end{equation}
where the sum is taken over all non-negative integers $m_{1},m_{2},...,m_{n}$
that verify the constraint $m_{1}+2m_{2}+3m_{3}+...+n\,m_{n}=n$,
and the derivatives of $\frac{f_{p}}{f_{\mu}}$ in the sum are with
respect to $\mu$.

For $n$ odd, that is, for an odd-order derivative, $n$ cannot be
the sum of even numbers. Therefore, there is at least one non-trivial
odd-indexed $m_{i}$ in the equation $m_{1}+2m_{2}+3m_{3}+...+n\,m_{n}=n$.
Then, for $n$ odd, in Eq.~(\ref{Appendix_eq:Faa_di_Bruno_ratio_EoS}),
we will always have in the product at least one term with an odd-order
derivative of $\mu$. Assuming that for all odd $k\leq n$, $\mu^{\left(k\right)}$$\left(0\right)$
vanish, Eq.~(\ref{Appendix_eq:Faa_di_Bruno_ratio_EoS}) implies~(\ref{Appendix_eq:Faa_di_Bruno_ratio_EoS_at_zero}).
\end{proof}

\begin{lem}
\label{Lemma:Vanishing_odd_derivative_p}Let $f\left(\mu,p\right)=0$
represent the equation of state of the fluid source, defining $p$
implicitly as a function of $\mu$. For $n\geq1$ odd, if $p\in\mathcal{C}^{n}\left(\mu\left(0\right)\right)$,
$\mu\in\mathcal{C}^{n}\left(0\right)$ and for all odd $k\in\left\{ 1,...,n\right\} $,
$\mu^{\left(k\right)}$$\left(0\right)$ vanish, then
\begin{equation}
\frac{d^{n}p}{dr^{n}}\left(0\right)=0\,.
\end{equation}
\end{lem}
The proof of Lemma~\ref{Lemma:Vanishing_odd_derivative_p} follows
the same reasoning of the proof of Lemma~\ref{Lemma:Vanishin_odd_derivative_EoS}.
\begin{lem}
\label{Lemma:Even_theta_zero_vanishin_odd_P}Let $n\geq0$ and the
matrix $\mathds{P}$ verify Eqs.~(\ref{Reformulation_eq:Power_series_P_matrix})
and (\ref{Reformulation_eq:Recurrence_relation_P}). If $\widetilde{\Theta}_{2k}=0$
for all\linebreak $k\in\left\{ 0,1,...,n\right\} $, then $\mathds{P}_{2n+1}=0$.
\end{lem}
\begin{proof}
The proof follows directly from the recurrence relation~(\ref{Reformulation_eq:Recurrence_relation_P}).
For $m=2n+1$, where $n\geq0$, we have
\begin{equation}
\mathds{P}_{2n+1}=\frac{1}{2n+1}\sum^{2n}_{i=0}\widetilde{\Theta}_{2n-i}\,\mathds{P}_{i}\,.\label{Appendix_eq:Sum_P_n_odd}
\end{equation}
If all $\widetilde{\Theta}_{2k}=0$ vanish for $0\leq k\leq n$, the
sum in Eq.~(\ref{Appendix_eq:Sum_P_n_odd}) contains only terms with
odd index $i$, that is, $\mathds{P}_{2n+1}$ depends only on the
coefficients $\mathds{P}_{i}$ with index $i$ odd smaller than $2n$.
If $\widetilde{\Theta}_{0}=0$, $\mathds{P}_{1}=0$. In turn, $\mathds{P}_{1}=0$
implies $\mathds{P}_{3}=0$, which implies $\mathds{P}_{5}=0$ and
so on.
\end{proof}

Lemma~\ref{Lemma:Even_theta_zero_vanishin_odd_P} follows from the
assumption that for a given value $2n+1$, all terms of the power
series of $\widetilde{\Theta}$ of even-order smaller than or equal
to $2n$ vanish. The following result provides sufficient conditions
for this to be verified.
\begin{lem}
\label{Lemma:mu_Theta_relation}Let $n\geq0$ and $\mu$ be real analytic
at $r=0$, such that $\mu^{\left(k\right)}$$\left(0\right)$ vanish,
for all odd $k\in\left\{ 1,...,2n+1\right\} $. Then, $\widetilde{\Theta}_{2n}=0$.
\end{lem}
\begin{proof}
To prove the lemma, we will explicitly find the coefficients $\widetilde{\Theta}_{m}$
of the power series of $\widetilde{\Theta}$ at $r=0$, in terms of
the derivatives of $\mu$. Since $\left(\widetilde{\Theta}\right)_{21}=r$
and $\left(\widetilde{\Theta}\right)_{22}=0$, their power series
are trivial. Then, we will focus specifically on the power series
of the $\left(1,1\right)$ and $\left(1,2\right)$ entries.

Assuming $\mu$ is real analytic at $r=0$ with power series $\mu\left(r\right)=\sum^{\infty}_{m=0}\frac{1}{m!}\mu^{\left(m\right)}\left(0\right)r^{m}$,
from Eq.~(\ref{Reformulation_eq:Mass_function_definition}) we have
\begin{equation}
2\left(1-\frac{2M}{r}\right)=2-2\sum^{\infty}_{m=0}\frac{\mu^{\left(m\right)}\left(0\right)}{m!\left(m+3\right)}\,r^{m+2}\,.
\end{equation}
Therefore,
\begin{equation}
2\left(1-\frac{2M}{r}\right)=\sum^{\infty}_{m=0}b_{m}\,r^{m}\,,
\end{equation}
where
\begin{equation}
\left\{ \begin{aligned}b_{0} & =2\\
b_{1} & =0\\
b_{m} & =-\frac{2}{\left(m+1\right)\left(m-2\right)!}\,\mu^{\left(m-2\right)}\left(0\right)\;,m\geq2
\end{aligned}
\right.\label{Appendix_eq:b_coeff}
\end{equation}
In particular, for $m\geq2$ even, the coefficients $b_{m}$ depend
only on the even-order derivatives of $\mu$ at zero.

Continuing, 
\begin{equation}
r\left(\mu-\frac{2M}{r^{3}}\right)=\sum^{\infty}_{m=0}\frac{m+2}{\left(m+3\right)m!}\,\mu^{\left(m\right)}\left(0\right)r^{m+1}\,.
\end{equation}
Therefore,
\begin{equation}
r\left(\mu-\frac{2M}{r^{3}}\right)=\sum^{\infty}_{m=0}a_{m}\,r^{m}\,,
\end{equation}
where
\begin{equation}
\left\{ \begin{aligned}a_{0} & =0\\
a_{m} & =\frac{m+1}{\left(m+2\right)\left(m-1\right)!}\,\mu^{\left(m-1\right)}\left(0\right)\;,m\geq1
\end{aligned}
\right.\label{Appendix_eq:a_coeff}
\end{equation}

Lastly,
\[
\frac{1}{r}\left(\mu-\frac{6M}{r^{3}}\right)=\sum^{\infty}_{m=1}\frac{m}{\left(m+3\right)m!}\,\mu^{\left(m\right)}\left(0\right)r^{m-1}\,,
\]
which can be written as
\begin{equation}
\frac{1}{r}\left(\mu-\frac{6M}{r^{3}}\right)=\sum^{\infty}_{m=0}\bar{a}_{m}\,r^{m}\,,
\end{equation}
with
\begin{equation}
\left\{ \begin{aligned}\bar{a}_{0} & =\frac{1}{4}\mu^{\left(1\right)}\left(0\right)\\
\bar{a}_{m} & =\frac{m+1}{\left(m+4\right)\left(m+1\right)!}\,\mu^{\left(m+1\right)}\left(0\right)\;,m\geq1
\end{aligned}
\right.\label{Appendix_eq:a_bar_coeff}
\end{equation}

From Eqs.~(\ref{Appendix_eq:a_coeff}) and (\ref{Appendix_eq:a_bar_coeff}),
for $m$ even, the coefficients $a_{m}$ and $\bar{a}_{m}$ depend
only on odd-order derivatives of $\mu$ at zero.

The division of power series can be expressed compactly using determinants.
Namely, for the $\left(1,1\right)$ entry we have
\begin{equation}
{\displaystyle \left(\widetilde{\Theta}\right)_{11}=\frac{r\left(\mu-\frac{2M}{r^{3}}\right)}{2\left(1-\frac{2M}{r}\right)}=\frac{{\displaystyle \sum^{\infty}_{m=0}a_{m}\,r^{m}}}{{\displaystyle \sum^{\infty}_{m=0}b_{m}\,r^{m}}}=\sum^{\infty}_{m=0}d_{m}\,r^{m}\,,}
\end{equation}
where
\begin{equation}
\left\{ \begin{aligned}d_{0} & =\frac{a_{0}}{b_{0}}\,,\\
d_{m} & =\frac{1}{b^{m+1}_{0}}\left|\begin{array}{ccccc}
a_{m} & b_{1} & b_{2} & \cdots & b_{m}\\
a_{m-1} & b_{0} & b_{1} & \cdots & b_{m-1}\\
a_{m-2} & 0 & b_{0} & \cdots & b_{m-2}\\
\vdots & \vdots & \vdots & \ddots & \vdots\\
a_{0} & 0 & 0 & \cdots & b_{0}
\end{array}\right|\,,\text{for }m\geq1
\end{aligned}
\right.\label{Appendix_eq:d_coeff}
\end{equation}
For $\left(1,2\right)$ entry,
\begin{equation}
{\displaystyle \left(\widetilde{\Theta}\right)_{12}={\displaystyle \frac{\mu-\frac{6M}{r^{3}}}{2r\left(1-\frac{2M}{r}\right)}}=\frac{{\displaystyle \sum^{\infty}_{m=0}\bar{a}_{m}\,r^{m}}}{{\displaystyle \sum^{\infty}_{m=0}b_{m}\,r^{m}}}=\sum^{\infty}_{m=0}\bar{d}_{m}\,r^{m}\,,}
\end{equation}
where
\begin{equation}
\left\{ \begin{aligned}\bar{d}_{0} & =\frac{\bar{a}_{0}}{b_{0}}\\
\bar{d}_{m} & =\frac{1}{b^{m+1}_{0}}\left|\begin{array}{ccccc}
\bar{a}_{m} & b_{1} & b_{2} & \cdots & b_{m}\\
\bar{a}_{m-1} & b_{0} & b_{1} & \cdots & b_{m-1}\\
\bar{a}_{m-2} & 0 & b_{0} & \cdots & b_{m-2}\\
\vdots & \vdots & \vdots & \ddots & \vdots\\
\bar{a}_{0} & 0 & 0 & \cdots & b_{0}
\end{array}\right|\,,\text{for }m\geq1
\end{aligned}
\right.\label{Appendix_eq:d_bar_coeff}
\end{equation}
From Eqs.~(\ref{Appendix_eq:d_coeff}) and (\ref{Appendix_eq:d_bar_coeff}),
$d_{0}=0$, and $\bar{d}_{0}=\frac{1}{8}\mu^{\left(1\right)}\left(0\right)$.
Then, $\widetilde{\Theta}_{0}=0$ , if $\mu^{\left(1\right)}\left(0\right)=0$.
For $m\geq2$ even, we will show that if all derivatives $\mu^{\left(k\right)}\left(0\right)$,
with odd $k\in\left\{ 1,...,m-1\right\} $, vanish, then $\left(\widetilde{\Theta}_{m}\right)_{11}$
is zero. The proof for $\left(\widetilde{\Theta}_{m}\right)_{12}$
follows similarly, assuming that all derivatives $\mu^{\left(k\right)}\left(0\right)$,
with odd $k\in\left\{ 1,...,m+1\right\} $, vanish.

We will prove by recurrence that for $m\geq2$ even, $d_{m}=0$, using
Laplace cofactor expansion for the determinant. To not have to worry
about the signs that follow from applying Laplace expansion, we will
consider the absolute value, $\left|d_{m}\right|$.

For an even $m\geq2$, $\left(\widetilde{\Theta}_{m}\right)_{11}$
is given by
\begin{equation}
d_{m}=\frac{1}{2^{m+1}}\left|\begin{array}{ccccc}
a_{m} & b_{1} & b_{2} & \cdots & b_{m}\\
a_{m-1} & b_{0} & b_{1} & \cdots & b_{m-1}\\
a_{m-2} & 0 & b_{0} & \cdots & b_{m-2}\\
\vdots & \vdots & \vdots & \ddots & \vdots\\
a_{0} & 0 & 0 & \cdots & b_{0}
\end{array}\right|\,,\label{Appendix_eq:d_matrix_general}
\end{equation}
where we have used $b_{0}=2$. If all derivatives $\mu^{\left(k\right)}\left(0\right)$,
with odd $k\leq m-1$, vanish, all $b_{k}$, and $a_{k+1}$, are zero.
Given the structure of the matrix in Eq.~(\ref{Appendix_eq:d_matrix_general}),
using Laplace expansion, the determinant is associated with the determinant
of the submatrix with the last line and last column removed from the
original matrix, such that
\begin{equation}
\left|d_{m}\right|=\left|\frac{1}{2^{m+1}}\det\left[\begin{array}{ccccc}
a_{m} & b_{1} & b_{2} & \cdots & b_{m}\\
a_{m-1} & b_{0} & b_{1} & \cdots & b_{m-1}\\
a_{m-2} & 0 & b_{0} & \cdots & b_{m-2}\\
\vdots & \vdots & \vdots & \ddots & \vdots\\
a_{0} & 0 & 0 & \cdots & b_{0}
\end{array}\right]\right|=\left|\frac{1}{2^{m}}\det\left[\begin{array}{ccccc}
a_{m} & b_{1} & b_{2} & \cdots & b_{m-1}\\
a_{m-1} & b_{0} & b_{1} & \cdots & b_{m-2}\\
a_{m-2} & 0 & b_{0} & \cdots & b_{m-3}\\
\vdots & \vdots & \vdots & \ddots & \vdots\\
a_{1} & 0 & 0 & \cdots & b_{0}
\end{array}\right]\right|\,.
\end{equation}
In turn, since $b_{1}=0$, the second column only has one non-trivial
entry, hence
\begin{equation}
\left|d_{m}\right|=\left|\frac{1}{2^{m-1}}\det\left[\begin{array}{ccccc}
a_{m} & b_{2} & b_{3} & \cdots & b_{m-1}\\
a_{m-2} & b_{0} & b_{1} & \cdots & b_{m-3}\\
a_{m-3} & 0 & b_{0} & \cdots & b_{m-4}\\
\vdots & \vdots & \vdots & \ddots & \vdots\\
a_{1} & 0 & 0 & \cdots & b_{0}
\end{array}\right]\right|\,.
\end{equation}
Since $b_{1}=b_{3}=0$, again the second column only has one non-trivial
entry, therefore
\begin{equation}
\left|d_{m}\right|=\left|\frac{1}{2^{m-2}}\det\left[\begin{array}{ccccc}
a_{m} & b_{2} & b_{4} & \cdots & b_{m-1}\\
a_{m-2} & b_{0} & b_{2} & \cdots & b_{m-3}\\
a_{m-4} & 0 & b_{0} & \cdots & b_{m-5}\\
\vdots & \vdots & \vdots & \ddots & \vdots\\
a_{1} & 0 & 0 & \cdots & b_{0}
\end{array}\right]\right|\,.
\end{equation}
Following this procedure, since all even indexed coefficients $a_{i}$
vanish, we will end up with a matrix of size $\frac{m}{2}\times\frac{m}{2}$
whose first column is composed of only zeros, hence the determinant
is zero. Curiously, the procedure starts and ends by considering a
submatrix with the last line and last column removed from a previous
matrix.

We have shown that, for $m$ even, $d_{m}=0$. The proof that $\bar{d}_{m}=0$,
for $m$ even, follows similarly. Then, for $m$ even, $\widetilde{\Theta}_{m}=0$.
\end{proof}

The last intermediate result that we will need concerns the structure
of the power series of $\frac{\mathcal{A}}{r\phi}$.
\begin{lem}
\label{Lemma:A_over_r_phi}Let $s$ and $u$ be a solution of Eqs.~(\ref{Reformulation_eq:Solution_s_u})--(\ref{Reformulation_eq:Theta_tilde}),
then
\begin{equation}
\left(\frac{\mathcal{A}}{r\phi}\right)\left(0\right)=0\,.
\end{equation}
 If $\mu$ is a real analytic function at $r=0$, such that for an
even $m\geq2$, $\mu^{\left(k\right)}$$\left(0\right)$ vanish, for
all odd $k\in\left\{ 1,...,m-1\right\} $, then,
\begin{equation}
\left(\frac{\mathcal{A}}{r\phi}\right)^{\left(k+1\right)}\left(0\right)=0\,.
\end{equation}
\end{lem}
\begin{proof}
Equation~(\ref{Reformulation_eq:Solution_s_u_A}) implies $\frac{2\mathcal{A}}{r\phi}=\frac{s}{u}$.
Then, it suffices to show that truncating the power series of $u$
at some even order, it contains only terms with even powers of $r$.

From Eqs.~(\ref{Reformulation_eq:Solution_s_u}) and (\ref{Reformulation_eq:Power_series_P_matrix})
we have
\begin{equation}
\left[\begin{array}{r}
s\\
u
\end{array}\right]=\left[\begin{array}{lr}
r & 0\\
0 & 1
\end{array}\right]\left(\mathds{P}_{0}+\mathds{P}_{1}\,r+\mathds{P}_{2}\,r^{2}+...\right)\left[\begin{aligned}c_{1}\\
c_{2}
\end{aligned}
\right]\,,
\end{equation}
hence
\begin{equation}
\begin{aligned}s & =\sum^{\infty}_{i=0}\left[c_{1}\left(\mathds{P}_{i}\right)_{11}+c_{2}\left(\mathds{P}_{i}\right)_{12}\right]r^{i+1}=\sum^{\infty}_{i=0}s_{i}\,r^{i+1}\,,\\
u & =\sum^{\infty}_{i=0}\left[c_{1}\left(\mathds{P}_{i}\right)_{21}+c_{2}\left(\mathds{P}_{i}\right)_{22}\right]r^{i}=\sum^{\infty}_{i=0}u_{i}\,r^{i}\,.
\end{aligned}
\end{equation}
In particular, $s_{0}=0$ and, for $i\geq1$, $s_{i}$ depends on
$\mathds{P}_{i-1}$, and $u_{i}$ depends on $\mathds{P}_{i}$. The
quotient $s/u$ can be written as a power series in terms of the power
series of $s$ and $u$, such that
\begin{equation}
\frac{s}{u}={\displaystyle \sum^{\infty}_{i=0}d_{i}\,r^{i}\,,}
\end{equation}
where
\begin{equation}
\left\{ \begin{aligned}d_{0} & =0\,,\\
d_{m} & =\frac{1}{u^{m+1}_{0}}\left|\begin{array}{ccccc}
s_{m} & u_{1} & u_{2} & \cdots & u_{m}\\
s_{m-1} & u_{0} & u_{1} & \cdots & u_{m-1}\\
s_{m-2} & 0 & u_{0} & \cdots & u_{m-2}\\
\vdots & \vdots & \vdots & \ddots & \vdots\\
s_{0} & 0 & 0 & \cdots & u_{0}
\end{array}\right|\,,\text{for }m\geq1
\end{aligned}
\right.
\end{equation}
In particular, since the zeroth order coefficient $d_{0}$ is identically
zero, it implies that $\left(\frac{2\mathcal{A}}{r\phi}\right)\left(0\right)=0$.
Moreover, for a given $m\geq1$, the coefficient $d_{m}$ depends
on the matrix coefficients $\mathds{P}_{i}$, with $i=\left\{ 0,1,...,m\right\} $.

For an even $m\geq2$, if the derivatives $\mu^{\left(k\right)}$$\left(0\right)$
vanish for all odd $k\in\left\{ 1,...,m-1\right\} $, Lemmas~\ref{Lemma:Even_theta_zero_vanishin_odd_P}
and \ref{Lemma:mu_Theta_relation} imply that all $\mathds{P}_{k}$
are identically zero. Therefore, all $u_{k}$ vanish and $\sum^{m}_{i=0}u_{i}\,r^{i}$
defines an even function of the radial coordinate. Naturally, this
implies $\sum^{m}_{i=0}s_{i}\,r^{i}$ defines an odd function, since
$s=du/dr$. Therefore,
\begin{equation}
\frac{{\displaystyle \sum^{m}_{i=0}s_{i}\,r^{i}}}{{\displaystyle \sum^{m}_{i=0}u_{i}\,r^{i}}}\,,
\end{equation}
is an odd function in the radial coordinate, in particular, all even-order
terms of its power series expansion about $r=0$ are identically zero.
\end{proof}

Lemmas~\ref{Lemma:Vanishin_odd_derivative_EoS}--\ref{Lemma:A_over_r_phi}
follow from assuming that all odd-order derivatives of $\mu$ of order
smaller than or equal to a particular $n$, vanish at $r=0$. The
last step in the proof of Theorem~\ref{Theorem:Even_power_series}
is to show that analytic solutions of Eqs.~(\ref{Reformulation_eq:Solution_s_u})--(\ref{Reformulation_eq:Theta_tilde})
together with Eq.~(\ref{General_Power_Series_eq:EoS_derivative_mu})
have exactly this property.

Taking the derivative of order $n$ of Eq.~(\ref{General_Power_Series_eq:EoS_derivative_mu})
yields
\begin{equation}
\left(\mu\right)^{\left(n+1\right)}\left(0\right)=2\sum^{n}_{k=0}\sum^{k}_{i=0}\left(\begin{array}{c}
n\\
k
\end{array}\right)\left(\begin{array}{c}
k\\
i
\end{array}\right)\left(\frac{f_{p}}{f_{\mu}}\right)^{\left(n-k\right)}\left(\mu+p\right)^{\left(k-i\right)}\left(\frac{\mathcal{A}}{r\phi}\right)^{\left(i\right)}\,,\label{Appendix_eq:derivative_mu_general}
\end{equation}
where
\begin{equation}
\left(\begin{array}{c}
n\\
k
\end{array}\right)=\frac{n!}{k!\left(n-k\right)!}\,,
\end{equation}
and all derivatives are taken with respect to $r$ and evaluated at
$r=0$.

For $n=0$, Lemma~\ref{Lemma:A_over_r_phi} states that $\left(\frac{\mathcal{A}}{r\phi}\right)\left(0\right)$.
Then, from Eq.~(\ref{Appendix_eq:derivative_mu_general}), $\mu^{\left(1\right)}\left(0\right)=0$.

For $n=2$, since $\mu^{\left(1\right)}\left(0\right)=0$, from Lemma~\ref{Lemma:A_over_r_phi},
the derivatives $\left(\frac{\mathcal{A}}{r\phi}\right)^{\left(i\right)}$,
with $i\in\left\{ 0,2\right\} $, vanish, such that in the sums in~(\ref{Appendix_eq:derivative_mu_general}),
only the terms with $i$ odd, $i=1$, may be nontrivial. However,
if $k$ is even, then $k-i$ is odd. Using Lemma~\ref{Lemma:Vanishing_odd_derivative_p},
the term $\left(\mu+p\right)^{\left(k-i\right)}$ vanishes. If $k$
is odd, $n-k$ is odd. From Lemma~\ref{Lemma:Vanishin_odd_derivative_EoS},
the term $\left(\frac{f_{p}}{f_{\mu}}\right)^{\left(n-k\right)}$
is zero. Therefore, $\mu^{\left(3\right)}\left(0\right)=0$.

The cases for $n=\left\{ 5,7,...\right\} $ follow similarly.

\end{document}